\documentclass[journal]{IEEEtran}
\usepackage{xcolor}

\usepackage{amsmath}
\usepackage{amssymb}
\usepackage{graphicx}
\usepackage{booktabs}
\usepackage{multirow}
\usepackage{wrapfig}
\usepackage{capt-of}
\usepackage{subcaption}
\usepackage{algorithm}
\usepackage{algorithmic}
\usepackage{enumitem}
\usepackage{epstopdf}
\usepackage[colorlinks=true, linkcolor=blue, citecolor=blue]{hyperref}

\newtheorem{lemma}{Lemma}

\allowdisplaybreaks[4]

\title{Aging-Aware Online Distributed Scheduling for Lifecycle Carbon
Reduction in Geo-Distributed Data Centers}

\author{
Junyu Lin,
Wenjie Liu,
Shunbo Lei,~\IEEEmembership{Senior Member, IEEE},
Wentian Lu,
Jianhui Wang,~\IEEEmembership{Member, IEEE},
Junhong Liu,~\IEEEmembership{Member, IEEE}%
\thanks{Junyu Lin, Wenjie Liu, Wentian Lu, and Jianhui Wang are with the
School of Mechanical and Electrical Engineering, Guangzhou University,
Guangzhou 510006, Guangdong, China
(e-mail: linjunyu20040910@gmail.com; lwj1993@gzhu.edu.cn;
luwentian@gzhu.edu.cn; jhwang@gzhu.edu.cn).}%
\thanks{Shunbo Lei is with the School of Science and Engineering,
The Chinese University of Hong Kong (Shenzhen), Longgang, Shenzhen,
Guangdong, China, and also with the Shenzhen Research Institute of
Big Data, Guangdong, China
(e-mail: leishunbo@cuhk.edu.cn).}%
\thanks{Junhong Liu is with the Rausser College of Natural Resources,
UC Berkeley, Berkeley, CA 94720, USA
(e-mail: junhongliu@berkeley.edu).}%
}

\begin{document}

\maketitle
\begin{abstract}
The rapid proliferation of data centers (DCs), driven by cloud computing and artificial intelligence (AI), has led to massive energy demand and carbon emissions, posing significant sustainability challenges. Carbon-aware optimization in geographically distributed data centers has been widely studied. Most existing approaches mainly focus on operational carbon emissions from server usage. However, existing literature often ignores workload-induced thermal stress, which accelerates nonlinear hardware degradation. This leads to more frequent server replacements and ultimately increases embodied carbon emissions\footnote{In this study, embodied carbon refers to the carbon emissions generated during server manufacturing and disposal, which are amortized into operational time slots based on hardware degradation.}. To address these limitations, we propose a comprehensive carbon life-cycle modeling framework for distributed data centers. Apart from operational carbon emissions, this work combines workload scheduling with a utilization-dependent exponential aging model to evaluate long-term carbon costs from server degradation. In order to solve the proposed optimization model in an online and privacy-preserving manner, an enhanced Lyapunov framework with time-varying queue shifting (TVQS) is first introduced to handle system uncertainties. Then, a zero-sum perturbation-based alternating direction method of multipliers (ZSP-ADMM) framework is developed to enable distributed coordination across geographically separated data centers
while protecting locally exchanged workload information. Simulation results demonstrate that the proposed approach achieves up to 13.0\% lower carbon emissions and 12.6\% lower operational costs compared with benchmarks.
\end{abstract}


\section{Introduction}

The rapid proliferation of data centers (DCs), driven by cloud computing and AI, has led to massive energy requirements and carbon emissions. DC power demand is projected to soar to nearly 2967~TWh by 2030, with its carbon footprint potentially contributing 8\% of global emissions \cite{katal2023energy, cao2022toward}. These trends raise significant concerns regarding electricity costs and environmental sustainability, making cost-efficient and carbon-aware operation an imperative research priority \cite{long2022review}.

To address these challenges, extensive research efforts have focused on real-time carbon-aware operational management in geographically distributed data centers \cite{cao2022toward}. Early strategies prioritized optimizing energy portfolios via Power Purchase Agreements and Renewable Energy Certificates \cite{ren2012carbon}. Subsequently, research pivoted toward operational flexibility, leveraging geographical load balancing and Lyapunov-based online scheduling to exploit spatio-temporal variations in carbon intensity and renewable energy availability \cite{zhou2013carbon, dou2017carbon, zhao2026lyapunov}. These concepts have reached industrial maturity through Google’s carbon-intelligent system using Virtual Capacity Curves \cite{radovanovic2022carbon}, with further refinements extending to dual-dimension task migration and edge-computing resource allocation \cite{yang2021carbon, gnibga2023latency}.

Existing thermal-aware hardware aging studies have shown that sustained high power density can accelerate physical wear in computing systems \cite{kumar2026advanced, 10450488}. Specifically, higher server utilization accelerates hardware aging, shortens server lifetime, and increases carbon emissions associated with hardware replacement \cite{11096406,9926126}. Field data from large-scale data centers indicate that regular hardware replacement is unavoidable due to inevitable component wear and failure \cite{vishwanath2010characterizing}. Such replacement requires continuous manufacturing of new devices and disposal of retired equipment. In particular, manufacturing high-performance components such as solid-state drives (SSDs) generates substantial embodied carbon emissions \cite{tannu2023dirty}. Although these manufacturing and disposal emissions account for over half of the total lifecycle carbon footprint in modern computing \cite{gupta2021chasing}, they are rarely integrated into dynamic data center scheduling. 

Most existing approaches, including carbon-aware workload scheduling methods \cite{lin2026carbon,ren2026grid,wang2025multi} and cloud-edge lifecycle management frameworks \cite{mekki2026energy}, focus primarily on energy efficiency or operational carbon emissions. They overlook heterogeneous degradation caused by varying server utilization levels. This oversight leads to underestimation of the actual lifecycle carbon footprint \cite{11673521,huang2025bim}. Although a recent study \cite{zhang2025carbon} considers both operational and embodied carbon, its embodied carbon accounting is primarily based on active server counts and accumulated operating time, without explicitly modeling utilization-dependent hardware degradation. Consequently, it may underestimate the embodied carbon emissions associated with accelerated hardware degradation under high utilization. Moreover, its offline scheduling method relies on future information about workload arrivals, renewable generation, and electricity prices, while its centralized coordination raises privacy concerns through information exchange among data centers.

Lyapunov optimization has been widely applied to online data center scheduling, including the joint coordination of computing workloads and battery energy \cite{huang2024online}, electricity procurement under renewable generation and price uncertainties \cite{gao2024online}, and cluster energy minimization under data redundancy constraints \cite{zhao2026lyapunov}. However, standard Lyapunov frameworks rely on static virtual queue shift parameters, preventing energy storage from adaptively responding to real-time fluctuations in renewable generation and electricity prices. On the other hand, alternating direction method of multipliers (ADMM)-based algorithms have been widely applied to distributed data center coordination, including computational-task coordination among multiple data center operators \cite{han2022coordinative} and joint workload and energy scheduling across regional peer-to-peer energy markets \cite{jin2025unlocking,chen2025spatial}. However, exchanging partial operational data in these ADMM schemes still exposes workload patterns and user privacy.

To address the aforementioned modeling and algorithmic limitations, this paper proposes a low-carbon online distributed operation framework for geographically distributed data centers. The main contributions are summarized as follows:

\textbf{1)} We develop a lifecycle carbon model for geo-distributed data centers that jointly accounts for operational carbon emissions and hardware-aging-induced embodied carbon emissions. The model links cluster-level server utilization to nonlinear lifetime degradation and calculates the resulting embodied carbon emissions. Simulation results show that the proposed framework reduces total carbon emissions by 13.0\% and operational costs by 12.6\% relative to the benchmark.

\textbf{2)} We formulate a long-term optimization problem that jointly minimizes economic costs and lifecycle carbon emissions, and develop a Lyapunov-based online method that transforms it into per-slot subproblems without requiring future system information. A time-varying queue shifting (TVQS) mechanism adjusts the battery virtual-queue shift according to real-time electricity prices and renewable generation. Simulation results show that the proposed mechanism enables the battery to adapt its charging and discharging decisions to variations in electricity prices and renewable availability.

\textbf{3)} We propose a zero-sum perturbation ADMM (ZSP-ADMM) framework
to enable distributed coordination across geographically separated
data centers without directly exposing their local workload migration
decisions. A structured perturbation mechanism is incorporated into
the ADMM information exchange, where coordinated zero-sum
perturbations preserve the global workload-balance structure while
obfuscating local migration signals. Simulation results demonstrate
that the proposed ZSP-ADMM achieves effective workload masking with
negligible computational overhead.

\section{Proposed DC Operation Framework}

\begin{figure}[h]
    \centering
    \includegraphics[width=\linewidth]{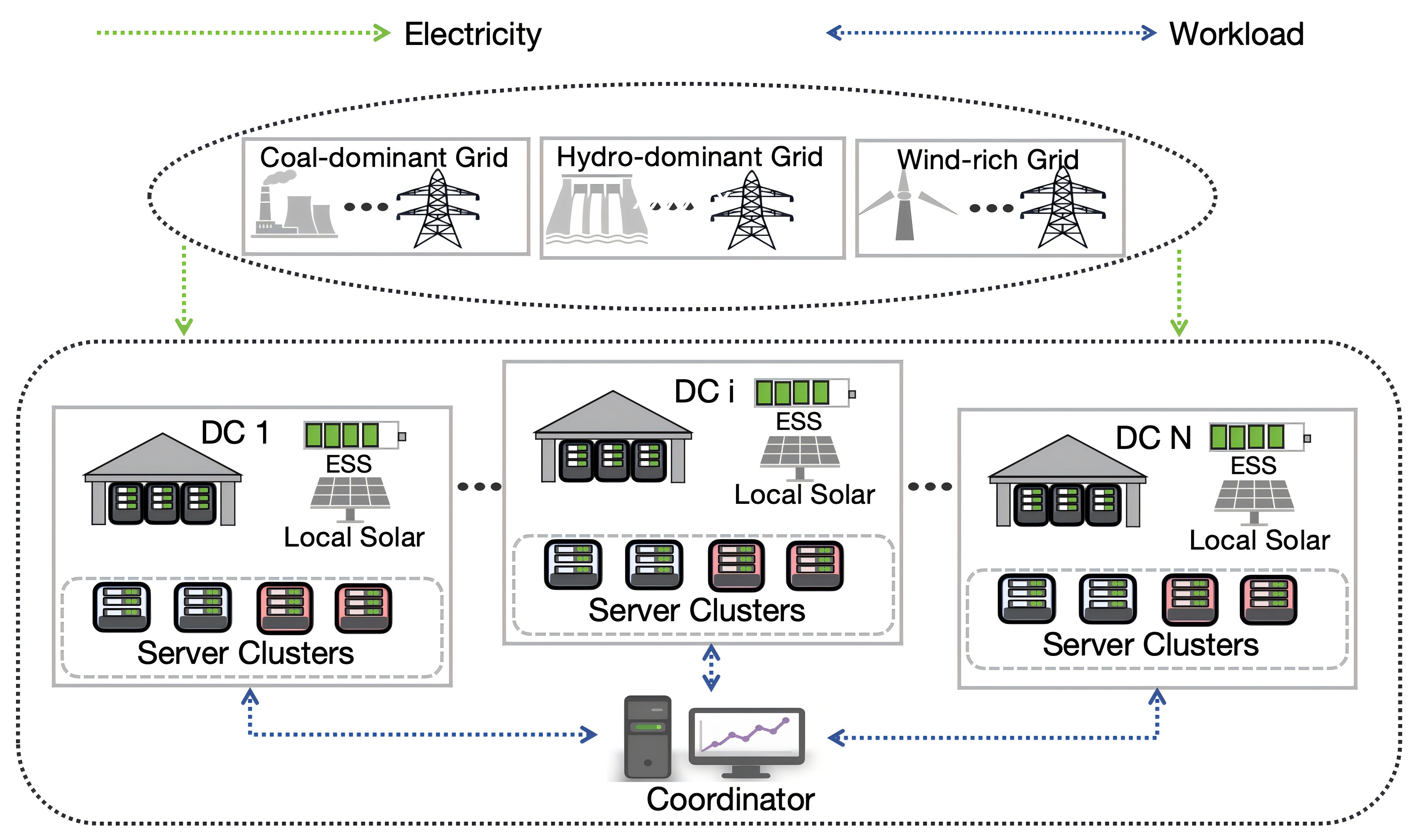}
    \caption{Structure of the data centers.}
    \label{fig_dc}
\end{figure}
We consider a geographically distributed DC network, as illustrated in Fig.~\ref{fig_dc}. Externally, each DC operates under distinct local electricity prices and grid carbon emission intensities depending on its regional power sources. Internally, each DC is equipped with local solar generation, energy storage systems (ESS), and heterogeneous server clusters that exhibit distinct processing rate, power profiles, and physical aging sensitivity factors.

Three coupled flows are modeled and jointly optimized. 1) Electricity flow: Each DC can utilize local renewable energy generation and an energy storage system (ESS), while purchasing electricity from the nearby utility grid and incurring both electricity cost and operational carbon emissions. 2) Workload flow: Interactive workloads can be migrated among DCs. 3) Degradation flow: Workload allocation among server clusters directly impacts hardware degradation, which dictates server replacement cycles and embodied carbon emissions.
\subsection*{A. Operational Cost Model}

\subsubsection*{1) Electricity Purchase Cost}

The electricity purchase cost of data center $i$ from the power grid at time slot $t$ is modeled as
\begin{equation}
   f^{g}_{i,t} = \lambda_{i,t}\, P_{g,i,t}\, \Delta t.
\end{equation}
Here, $P_{g,i,t}$ represents the active power drawn by data center $i$ from the grid at time $t$, $\lambda_{i,t}$ denotes the electricity price at the bus where DC $i$ is connected, and $\Delta t$ denotes the duration of each time slot.

\subsubsection*{2) Energy Storage System Cost}

Battery cycling leads to degradation and an associated cost. We model the degradation-related cost of the battery in DC $i$ at time slot $t$ as
\begin{equation}
    f^{b}_{i,t} = \psi_{deg}\big(P_{b,i,t}^c + P_{b,i,t}^d\big)\,\Delta t,
\end{equation}
where $P_{b,i,t}^c$ and $P_{b,i,t}^d$ denote the charging and discharging power of the battery, respectively, and $\psi_{deg}$ is a degradation coefficient that converts energy throughput into an equivalent monetary cost.

\subsubsection*{3) Transfer Cost}

Interactive workload migration across data centers incurs operational costs, such as bandwidth charges or routing overhead. We assume this cost is borne by the sending data center. For each DC $i$, the migration cost incurred at time slot $t$ is defined as:
\begin{equation}
    f^{\mathrm{trans}}_{i,t} = c^{\mathrm{trans}}_i [m^{I}_{i,t}]^+, \quad \forall i\in\mathcal{I}, \forall t,
\end{equation}
where $c^{\mathrm{trans}}_i$ is the unit transfer cost for DC $i$, and $[m^{I}_{i,t}]^+ = \max(m^{I}_{i,t}, 0)$ denotes the net outbound interactive workload migrated at time $t$.

\par\indent
\textit{4) Operational Carbon Cost:}\quad

Let $\alpha_C$ denote the carbon price per unit emission, and let $\rho^{oc}_{i,t}$ denote the carbon emission intensity of grid electricity at the location of DC $i$ and time $t$, measured in kg CO$_2$ per unit energy. The operational carbon cost incurred at time slot $t$ is modeled as
\begin{equation}
    f^{c,o}_{i,t}
    =
    \alpha_C\, \rho^{oc}_{i,t}\,
    P_{g,i,t}\,\Delta t,
    \qquad \forall i,t.
\end{equation}

\subsection*{B. Degradation-Informed Embodied Carbon Cost Model}

\begin{figure}[h]
    \centering
    \includegraphics[width=\linewidth]{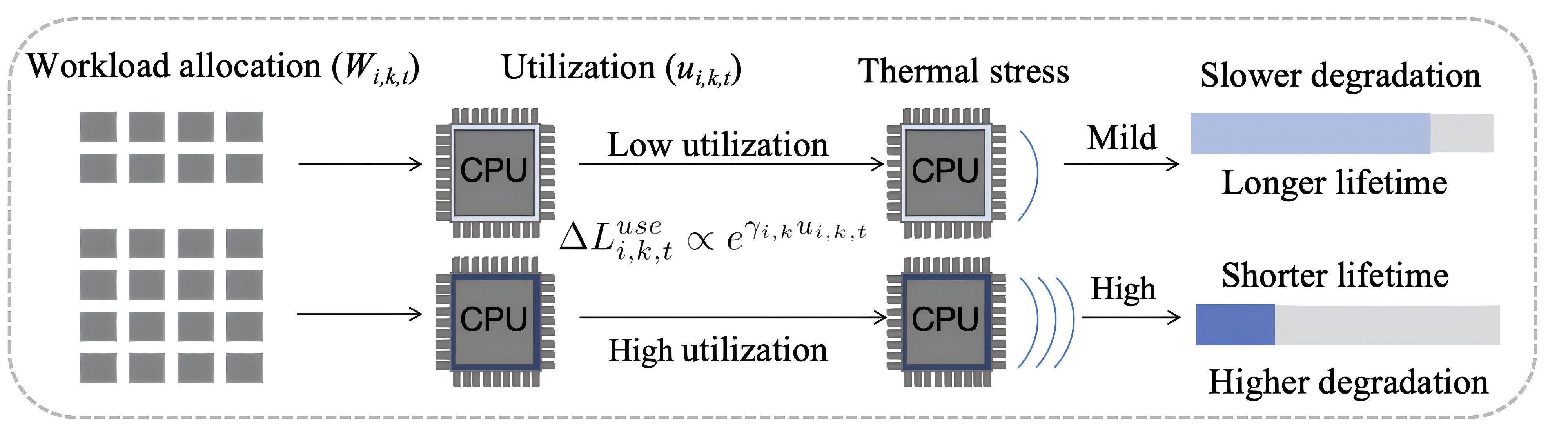}
    \caption{Illustration of utilization-induced server aging and lifetime consumption.}
    \label{fig_sever}
\end{figure}

Fig.~\ref{fig_sever} illustrates how workload allocation at the server-cluster level
is translated into server utilization and subsequently affects thermal stress
and lifetime consumption.

Conventional amortization models allocate embodied carbon uniformly over the expected hardware lifetime \cite{hewage2025aging}. However, physical hardware aging is highly load-dependent \cite{liu2019service}. High resource utilization elevates power consumption and processor temperature, thereby accelerating physical hardware aging \cite{hewage2025aging, dowling2025regulating}. This thermal stress shortens the operational lifespan of servers, necessitating more frequent hardware replacements and thus driving up the overall embodied carbon footprint. Accordingly, we develop a framework that links workload utilization to physical server aging and the embodied carbon cost induced by hardware replacement.

We define the resource utilization level $u_{i,k,t}\in[0,1]$ for
cluster $(i,k)$ at time slot $t$ as the ratio of the assigned workload
to the total available processing capacity:
\begin{equation}
    u_{i,k,t}
    \triangleq
    \frac{W_{i,k,t}}
    {N^{A}_{i,k,t}\,S^{\mathrm{rate}}_{i,k}},
    \qquad \forall i,k,t,
\end{equation}
where $S^{\mathrm{rate}}_{i,k}$ denotes the processing rate of a single server in cluster $(i,k)$, and $N^{A}_{i,k,t}$ represents the number of active servers. This definition captures the load-dependent nature of server aging, whereby a higher execution load results in a faster hardware aging rate \cite{liu2019service,leverich2014reconciling}.

\subsubsection*{1) Utilization-induced Aging Model}

To mathematically map the load-dependent physical wear to the macroscopic lifespan reduction, an exponential aging model is adopted. Let $L^{\mathrm{use}}_{i,k,t}$ represent the cumulative lifetime consumed by cluster $(i,k)$ up to time slot $t$. The lifetime dynamics evolve as:
\begin{equation}
    L^{\mathrm{use}}_{i,k,t+1}
    =
    L^{\mathrm{use}}_{i,k,t}
    +
    \underbrace{
    \eta_{i,k}\,
    \exp\!\big(\gamma_{i,k}\, u_{i,k,t}\big)\,
    \Delta t
    }_{\text{aging increment}},
\end{equation}
where $\eta_{i,k}>0$ is the intrinsic aging coefficient, and $\gamma_{i,k}>0$ denotes the aging sensitivity factor.

\subsubsection*{2) Embodied Carbon Accounting}

As server lifetime is consumed during operation, the manufacturing
embodied carbon is amortized in proportion to the corresponding lifetime
consumption. Therefore, sustained high utilization accelerates lifetime
consumption and consequently increases the embodied carbon attributed
to the current operation.

The fractional lifetime consumption during time slot $t$ is calculated as:
\begin{equation}
\Delta \ell_{i,k,t}
=
\frac{
\eta_{i,k}
\exp(\gamma_{i,k} u_{i,k,t})
\Delta t
}{
L^{\mathrm{life}}_{i,k}
},
\end{equation}
where $L^{\mathrm{life}}_{i,k}$ represents the nominal design lifetime under base operating conditions. The corresponding instantaneous embodied carbon expenditure is:
\begin{equation}
E^{\mathrm{emb}}_{i,k,t}
=
C^{M}_{i,k}\Delta \ell_{i,k,t},
\end{equation}
where $C^{M}_{i,k}$ denotes the manufacturing embodied carbon per server in cluster $(i,k)$.

\subsubsection*{3) Aging-Amortized Embodied Carbon Cost}

To account for the degradation-induced embodied carbon in the operational optimization loop, the embodied carbon emission is aggregated across all active servers and converted into an equivalent monetary cost via the carbon price. Specifically, the total embodied carbon cost for data center $i$ at time slot $t$ is obtained by multiplying the single-server embodied carbon $E^{\mathrm{emb}}_{i,k,t}$ by the active server count $N^A_{i,k,t}$ and the carbon price $\alpha_C$:
\begin{equation}
\begin{split}
    f^{c,m}_{i,t}
    &=
    \alpha_C
    \sum_{k \in \mathcal{K}_i}
    N^A_{i,k,t} E^{\mathrm{emb}}_{i,k,t} \\
    &=
    \alpha_C
    \sum_{k \in \mathcal{K}_i}
    N^A_{i,k,t}
    \left(
    C^{M}_{i,k}
    \frac{
    \eta_{i,k}
    \exp(\gamma_{i,k} u_{i,k,t})
    \Delta t
    }{
    L^{\mathrm{life}}_{i,k}
    }
    \right),
    \quad \forall i,t.
\end{split}
\end{equation}

\enlargethispage{1\baselineskip}

By combining the operational and embodied carbon, the total carbon-related financial objective for data center $i$ at time slot $t$ is established as:
\begin{equation}
    f^{c}_{i,t} = f^{c,o}_{i,t} + f^{c,m}_{i,t}, \qquad \forall i,t.
\end{equation}

\subsection*{C. Workload Migration Model}

We consider two workload types: interactive workloads (IWs) and batch workloads (BWs) \cite{ye2022joint}.
\begin{itemize}
    \item IWs: latency-critical tasks that must be completed within each time slot and can be migrated across data centers.
    \item BWs: delay-tolerant tasks processed locally due to high data transfer overhead.
\end{itemize}

\subsubsection*{1) Interactive Workload (IW) Modeling}

Let $w^{I,\mathrm{pre}}_{i,t}$ denote the local incoming IW demand at DC $i$ at time slot $t$, and let $m^{I}_{i,t}$ represent the net migrated IW volume from DC $i$. The IW processed at DC $i$ is formulated as:
\begin{equation}\label{eq:IW_local_balance}
w^I_{i,t} = w^{I,\mathrm{pre}}_{i,t} - m^{I}_{i,t}, \quad \forall i\in\mathcal{I}, \forall t,
\end{equation}
where $m^{I}_{i,t} > 0$ indicates net outbound workload migration (DC $i$ exports tasks), and $m^{I}_{i,t} < 0$ denotes net inbound workload absorption (DC $i$ imports tasks).

The net migrated workloads across all data centers must sum to zero at each time slot:
\begin{equation}\label{eq:workload_balance}
\sum_{i\in\mathcal{I}} m_{i,t}^I = 0, \quad \forall t.
\end{equation}

To ensure physical feasibility, the net migrated workload is bounded by the gateway bandwidth and local processing capacity limits:
\begin{equation}\label{eq:bandwidth_limit}
-\Delta m^{I,\mathrm{limit}}_{i,\mathrm{in}} \le m^{I}_{i,t} \le \min\left(w^{I,\mathrm{pre}}_{i,t}, \, \Delta m^{I,\mathrm{limit}}_{i,\mathrm{out}}\right), \quad \forall i\in\mathcal{I}, \forall t,
\end{equation}
where $\Delta m^{I,\mathrm{limit}}_{i,\mathrm{in}}$ and $\Delta m^{I,\mathrm{limit}}_{i,\mathrm{out}}$ denote the maximum allowable inbound and outbound migration capacity thresholds for DC $i$, respectively.

The workload processed locally at DC $i$ is further allocated to its clusters. Let $w^{I}_{i,k,t}$ denote the interactive workload processed by cluster $k$ in DC $i$, which satisfies:
\begin{equation}\label{eq:IW_cluster_alloc}
w^{I}_{i,t} = \sum_{k\in\mathcal{K}_i} w^{I}_{i,k,t}, \quad \forall i\in\mathcal{I}, \forall t.
\end{equation}
\subsubsection*{2) Batch Workload (BW) Modeling}

The BW queue evolves according to the standard discrete-time queueing dynamics:
\begin{equation}\label{eq:batch_queue}
    S^{D}_{i,t+1}
    =
    \max\!\Big\{
    S^{D}_{i,t}
    -
    w^{D}_{i,t},
    \,0
    \Big\}
    +
    A^{D}_{i,t},
    \qquad \forall i,t,
\end{equation}
where $S^{D}_{i,t}$ denotes the backlog of BW at DC $i$ at time slot $t$, and $A^{D}_{i,t}$ represents the newly arrived BW from local users during slot $t$. The total BW service $w^{D}_{i,t}$ provided by DC $i$ is formulated as
\begin{equation}\label{eq:BW_cluster_alloc}
    w^{D}_{i,t}
    =
    \sum_{k\in\mathcal{K}_i} w^{D}_{i,k,t},
    \qquad \forall i,t,
\end{equation}
where $w^{D}_{i,k,t}$ is the amount of BW processed by cluster $k$ in DC $i$.
Non-negativity constraints ensure feasibility:
\begin{equation}\label{eq:BW_non_negativity}
    S^{D}_{i,t} \ge 0,\quad
    w^{D}_{i,k,t} \ge 0,
    \qquad \forall i,k,t.
\end{equation}

Due to the delay-tolerant nature of batch workloads, queue stability
is adopted as the performance criterion for BW scheduling \cite{neely2010stochastic}. Accordingly,
the BW queue at each DC is required to satisfy the following stability condition:
\begin{equation}
\lim_{T\to\infty}
\frac{1}{T}
\sum_{t=0}^{T-1}
\mathbb{E}\!\left[S^D_{i,t}\right]
<\infty,
\qquad \forall i.
\label{eq:BW_strong_stability}
\end{equation}
\subsection*{D. Service Level Agreement (SLA) Constraint for Workload Processing}

The total workload assigned to cluster $(i,k)$ at time slot $t$ consists of
both interactive workload and delay-tolerant batch workload, given by
\begin{equation}\label{eq:total_workload_def}
    W_{i,k,t}
    =
    w^{I}_{i,k,t}
    +
    w^{D}_{i,k,t},
    \qquad \forall i,k,t.
\end{equation}

\noindent\textit{1) Delay Constraint}\quad

For cluster $(i,k)$ at time slot $t$, the workload execution delay is mainly
determined by the server utilization and available computing capacity. Using
a  M/M/1-type approximation based on aggregate service capacity, the
execution delay is given by
\begin{equation}\label{eq:SLA_delay_limit}
    t^{\mathrm{delay}}_{i,k,t}
    =
    \frac{1}
    {\big(1-u_{i,k,t}\big)\,S^{\mathrm{rate}}_{i,k}}
    \le
    \bar{t}^{\mathrm{delay}}_{i,k},
    \qquad \forall i,k,t.
\end{equation}
which provides a tractable delay approximation as commonly adopted in prior
work \cite{chen2023multi}. $\bar{t}^{\mathrm{delay}}_{i,k}$ denotes the
maximum allowable delay specified by the SLA.

\noindent\textit{2) Active Server Capacity Constraint}

The number of active servers in each cluster is constrained by the installed
capacity:
\begin{equation}\label{eq:active_server_limit}
    N^{\min}_{i,k} \le N^{A}_{i,k,t} \le N^{\max}_{i,k},
    \qquad \forall i,k,t,
\end{equation}
where $N^{\min}_{i,k}>0$ and $N^{\max}_{i,k}$ denote the minimum and maximum numbers of active servers in cluster $(i,k)$, respectively.
\subsection*{E. Computing Power Consumption and Thermal Dynamics Model}

\subsubsection*{1) Computing Power Consumption}
The power consumption of cluster $(i,k)$ at time $t$ is formulated as:
\begin{equation}\label{eq:IT_power_cluster}
    P^{\mathrm{IT}}_{i,k,t}
    = 
    P^{\mathrm{idle}}_{i,k} N^{A}_{i,k,t} 
    + 
    u_{i,k,t} N^{A}_{i,k,t} \big(P^{\mathrm{peak}}_{i,k} - P^{\mathrm{idle}}_{i,k}\big), \quad \forall i,k,t,
\end{equation}
where $P^{\mathrm{idle}}_{i,k}$ and $P^{\mathrm{peak}}_{i,k}$ denote the idle and peak power consumption of a single server, respectively. 

The total computing power consumption of data center $i$ is aggregated as:
\begin{equation}\label{eq:IT_power_aggregate}
    P^{\mathrm{IT}}_{i,t} = \sum_{k\in\mathcal{K}_i} P^{\mathrm{IT}}_{i,k,t}, \quad \forall i,t.
\end{equation}

\subsubsection*{2) Thermal Dynamics and Cooling}
The heat generated by server operation affects the indoor temperature evolution, which must be managed by the cooling system. Let $\tau_{i,t}$ denote the ambient temperature of data center $i$ at time $t$. The discrete-time thermal model is formulated as:
\begin{equation}\label{eq:thermal_dynamic}
    \tau_{i, t+1} = \tau_{i, t} + \kappa_i^B P^{\mathrm{IT}}_{i, t} - \kappa_i^C P^{\mathrm{cool}}_{i, t} - \beta_{i, t}, \quad \forall i,t,
\end{equation}
where $P^{\mathrm{cool}}_{i, t}$ is the cooling power. The parameters $\kappa_i^B$ and $\kappa_i^C$ represent the heat conversion coefficient and the cooling efficiency, respectively, while $\beta_{i, t}$ accounts for ambient heat dissipation. 

Constraint \eqref{eq:temp_limits} enforces the temperature limits for each data center:
\begin{equation}\label{eq:temp_limits}
    \underline{T}_i \le \tau_{i, t} \le \overline{T}_i, \quad \forall i,t.
\end{equation}
\subsection*{F. Power Management Constraints}

This subsection summarizes the constraints associated with power flows inside each DC. 

\noindent\textit{1) Electricity Purchase from the Grid}

The grid power drawn by each DC is bounded as
\begin{equation}\label{eq:grid_power_limit}
   0 \le P_{g,i,t} \le P_{g,i,t}^{\max},
\end{equation}
where $P_{g,i,t}^{\max}$ denotes the maximum grid power intake of data center $i$.

\noindent\textit{2) Battery Energy Storage System}

Charging and discharging power are limited by converter and battery ratings:
\begin{equation}\label{eq:battery_charge_limit}
    0 \le P_{b,i,t}^c \le P_{b,i,t}^{c,\max},
\end{equation}
\begin{equation}\label{eq:battery_discharge_limit} 
    0 \le P_{b,i,t}^d \le P_{b,i,t}^{d,\max},
\end{equation}
where $P_{b,i,t}^{c,\max}$ and $P_{b,i,t}^{d,\max}$ denote the maximum charging and discharging power.

The battery energy level evolves as
\begin{equation}\label{eq:battery_soc}
    E_{b,i,t+1}
    =
    k_i E_{b,i,t}
    + P_{b,i,t}^c\,\Delta t\,\eta_c
    - P_{b,i,t}^d\,\frac{\Delta t}{\eta_d},
\end{equation}
where $k_i \in (0,1]$ is the energy retention coefficient, and $\eta_c$, $\eta_d$ are charging/discharging efficiencies.

The stored energy is constrained by
\begin{equation}\label{eq:battery_capacity_limit}
   E_{b,i}^{\min}\le E_{b,i,t}\le E_{b,i}^{\max},
\qquad \forall i,t.
\end{equation}
and the inter-temporal variation is limited by
\begin{equation}
E_{b,i,t} - \Delta E_{b,i} \le E_{b,i,t+1} \le E_{b,i,t} + \Delta E_{b,i},
\end{equation}
where $\Delta E_{b,i}$ represents the ramp rate limit for battery storage.

\noindent\textit{3) DC Power Balance}

For each data center and time slot, the electrical supply must match demand:
\begin{equation}\label{eq:energy_balance}
    P_{g,i,t} + P_{b,i,t}^d + P_{r,i,t}
    =
    P^{\mathrm{IT}}_{i,t} + P^{\mathrm{cool}}_{i,t} + P_{b,i,t}^c,
    \qquad \forall i,t.
\end{equation}
Here, $P_{r,i,t}$ is the local renewable generation, and $P^{\mathrm{IT}}_{i,t}$ denotes the total computing power consumption of DC $i$.
\subsection*{G. Joint Cost and Carbon Optimization Problem}

Here, $i\in\mathcal{I}$, $k\in\mathcal{K}_i$, and
$t\in\mathcal{T}$ index the data centers, server clusters of DC $i$,
and time slots, respectively. Let \(\Phi_t\) define the control variables of all data centers at time slot \(t\).
\[
\Phi_t = \left\{
\begin{aligned}
    &m^{I}_{i,t},\, N^{A}_{i,k,t},\, w^{D}_{i,k,t},\, w^{I}_{i,k,t},\, P_{g,i,t}, \\
    &P^{\mathrm{cool}}_{i,t},\, P^{c}_{b,i,t},\, P^{d}_{b,i,t}
\end{aligned}
\right\}.
\]
The instantaneous operation cost $f(\Phi_t)$ incurred at time slot $t$ is defined as:
\begin{equation}
f(\Phi_t) \triangleq \sum_{i \in \mathcal{I}} \big( f^{g}_{i,t} + f^{b}_{i,t} + f^{\mathrm{trans}}_{i,t}+f^{c}_{i,t}  \big).
\end{equation}

By aggregating all the constraints presented in Sections~C--F, the joint cost and carbon optimization problem is formulated as follows:

{
\begin{flalign*}
&\textbf{P1:}\quad
\min_{\{\Phi_t\}_{t\in\mathcal T}}
\ \sum_{t\in\mathcal T} f(\Phi_t)
&&\\[0.5mm]
&\hspace{3.5em}\text{\rm s.t.}\quad
\eqref{eq:IW_local_balance}\text{--}\eqref{eq:BW_strong_stability},
\ \eqref{eq:total_workload_def}\text{--}\eqref{eq:energy_balance},
\quad \forall t\in\mathcal T.
&&
\end{flalign*}
}
\section{Solution Method}

Directly solving \textbf{P1} in a centralized offline manner is impractical due to two fundamental structural challenges:

First, problem \textbf{P1} exhibits strong inter-temporal coupling across scheduling horizons through physical state transitions, including \eqref{eq:batch_queue}, \eqref{eq:thermal_dynamic}, and \eqref{eq:battery_soc}. Offline optimization over the operating horizon requires prior knowledge of future renewable generation, electricity prices, and workload arrivals, which are uncertain. To address these temporal dependencies without future predictions, we design an online Lyapunov optimization framework to decouple the long-term problem into per-slot subproblems.

Second, the workload-balance constraint in \eqref{eq:workload_balance} couples migration decisions across data centers. A centralized solution requires sharing local operational data, raising privacy concerns. We therefore develop a zero-sum perturbation ADMM (ZSP-ADMM) algorithm that decomposes the problem into local subproblems. Structured zero-sum perturbations are applied to the exchanged migration signals, masking local decisions while preserving the global workload balance.

\subsection{Unified Online Optimization Framework}

To enable the real-time implementation, we reformulate the long-term
problem \textbf{P1} using Lyapunov optimization into a sequence of per-slot
deterministic optimization problems.

\subsubsection{Virtual Queue Modeling and Time-Varying Queue Shifting}
To handle inter-temporal constraints without future predictions, we introduce three queue-based state variables:
\begin{itemize}
    \item \textbf{Batch Workload Queue:} The backlog $S^{D}_{i,t}$ tracks delay-tolerant tasks, evolving according to $S^{D}_{i,t+1} = \max\{S^{D}_{i,t} - w^{D}_{i,t}, 0\} + A^{D}_{i,t}$.
    \item \textbf{Thermal Virtual Queue:} To enforce operational temperature limits, we define $\tilde{\tau}_{i,t} = \tau_{i,t} - T^{\mathrm{ref}}_i$, capturing the temperature deviation from the baseline $T^{\mathrm{ref}}_i$.
    \item \textbf{Battery Virtual Queue:} To manage battery energy bounds within the Lyapunov framework \cite{neely2010stochastic}, we define the shifted virtual queue $\tilde{E}_{i,t} = E_{b,i,t} + \varepsilon_{i,t}$, where $E_{b,i,t}$ represents the battery energy level, and $\varepsilon_{i,t}$ is a shift parameter determined by the time-varying queue shifting (TVQS) strategy below.

\end{itemize}

\textbf{Time-Varying Queue Shifting Mechanism:}
Standard Lyapunov methods use a fixed battery virtual-queue shift derived from offline worst-case bounds, which may lead to conservative battery scheduling under time-varying operating conditions. We therefore develop a time-varying queue shifting (TVQS) mechanism that adjusts the battery virtual-queue shift according to real-time renewable generation and electricity prices, thereby adapting the charging and discharging incentives to current operating conditions.

Specifically, the operational incentive factor and the corresponding virtual-queue shift are defined as: 
\allowdisplaybreaks

\begin{equation}
\theta_{i,t}
=
\min\left\{
1,\,
\beta_\theta
\left(
\frac{P^{\mathrm{r}}_{i,t}}
{P^{\mathrm{r,rated}}_i + P^{\mathrm{r}}_{i,t}}
+
\frac{\lambda_i^{\max}-\lambda_{i,t}}
{\lambda_i^{\max}}
\right)
\right\},
\end{equation}
\begin{equation}
\varepsilon_{i,t}=-E^{\max}_{b,i}\theta_{i,t}.
\end{equation}
Here, $\theta_{i,t}\in[0,1]$ is the operational incentive factor, and $\beta_\theta>0$ controls the overall strength of the queue-shifting mechanism. $P^{\mathrm{r}}_{i,t}$ and $P^{\mathrm{r,rated}}_i$ denote the current and rated renewable generation, respectively; $\lambda_{i,t}$ and $\lambda_i^{\max}$ denote the current electricity price and its maximum threshold, where $\lambda_{i,t}\leq\lambda_i^{\max}$; and $E^{\max}_{b,i}$ denotes the maximum battery capacity. Accordingly, $\varepsilon_{i,t}\in[-E^{\max}_{b,i},0]$.

The value of $\theta_{i,t}$ increases when renewable generation is abundant or electricity prices are low. Consequently, $\varepsilon_{i,t}$ becomes more negative, lowering the battery virtual-queue value and strengthening the charging incentive. Conversely, when renewable generation decreases or electricity prices rise, $\theta_{i,t}$ decreases and $\varepsilon_{i,t}$ becomes less negative, raising the battery virtual-queue value and weakening the charging incentive.

\subsubsection{Lyapunov Function Formulation and Drift Bounding}
Let $\boldsymbol{\Theta}_t = [S^{D}_{i,t}, \tilde{E}_{i,t}, \tilde{\tau}_{i,t}]_{i \in \mathcal{I}}$ denote the concatenated system state vector. The unified Lyapunov function is defined as:
\begin{equation}
    L(\boldsymbol{\Theta}_t) = \frac{1}{2} \sum_{i \in \mathcal{I}} \left( \alpha_i (S^{D}_{i,t})^2 + \delta_i \tilde{E}_{i,t}^2 + \nu_i \tilde{\tau}_{i,t}^2 \right),
\end{equation}
where $\alpha_i, \delta_i, \nu_i > 0$ are positive weighting parameters.

The one-slot conditional Lyapunov drift is defined as $\Delta(\boldsymbol{\Theta}_t) = \mathbb{E}[L(\boldsymbol{\Theta}_{t+1}) - L(\boldsymbol{\Theta}_t) \mid \boldsymbol{\Theta}_t]$. 
The upper bound of the one-slot conditional Lyapunov drift is given in the following lemma.

\begin{lemma}\label{lemma_drift_bound}
Under any feasible control action and system state realization, the conditional one-slot Lyapunov drift is upper bounded by:
\begingroup
\small

\setlength{\jot}{1pt}

\begin{equation}
\begin{aligned}
\Delta(\boldsymbol{\Theta}_t)
\leq {}&
C_0
+\sum_{i\in\mathcal{I}}
\Bigg\{
\alpha_i S^D_{i,t}
\big(A^D_{i,t}-w^D_{i,t}\big)
+\nu_i\tilde{\tau}_{i,t}
\Big(
\kappa_i^B P^{\mathrm{IT}}_{i,t}
\\[-1.2mm]
&
\quad
-\kappa_i^C P^{\mathrm{cool}}_{i,t}
\Big)
+\delta_i\omega_{i,t}
\Big[
k_i\tilde{E}_{i,t}
+\tfrac{1}{2}\omega_{i,t}
+(1-k_i)\varepsilon_{i,t}
\Big]
\Bigg\},
\end{aligned}
\end{equation}
\endgroup
where $\omega_{i,t} = P^c_{b,i,t}\Delta t\eta_c - P^d_{b,i,t}\Delta t/\eta_d$ denotes the net battery energy variation at time slot $t$, and $C_0$ is a finite positive constant bounding the remainder terms arising from the queue dynamics and the time-varying shift variation.
\end{lemma}

\par\emph{Proof.}
The detailed derivation is provided in Appendix~\ref{appendix_A}.

\subsubsection{Per-slot Online Optimization Problem}
To minimize the operational cost while stabilizing all queues, we
adopt the drift-plus-penalty framework by incorporating the instantaneous
operational cost $f(\Phi_t)$ alongside the drift bound established in
Lemma~\ref{lemma_drift_bound}. Let $V>0$ denote the Lyapunov control
parameter that balances queue stability and instantaneous operational
cost. By taking the upper bound of
$\Delta(\boldsymbol{\Theta}_t)
+V\mathbb{E}[f(\Phi_t)\mid\boldsymbol{\Theta}_t]$
and neglecting terms independent of the control actions, the per-slot
deterministic optimization problem at time $t$ is formulated as:
\par
{\small
\setlength{\jot}{1pt}
\begin{equation}
\label{eq_p2_obj}
\begin{gathered}
\begin{aligned}
\textbf{P2:}\quad
\min_{\Phi_t}\quad
&
Vf(\Phi_t)
+\sum_{i\in\mathcal I}
\Big[
\delta_i\omega_{i,t}
\big(k_i\tilde E_{i,t}
+\tfrac12\omega_{i,t}
+(1-k_i)\varepsilon_{i,t}\big)
\\[-1.2mm]
&
\qquad
-\alpha_i S^D_{i,t}w^D_{i,t}
+\nu_i\tilde\tau_{i,t}
\big(\kappa_i^B P^{\mathrm{IT}}_{i,t}
-\kappa_i^C P^{\mathrm{cool}}_{i,t}\big)
\Big]
\end{aligned}
\\[0.5mm]
\hspace{-4.5em}
\begin{aligned}
\text{\rm s.t.}\quad &
\eqref{eq:IW_local_balance}\text{--}\eqref{eq:bandwidth_limit},
\eqref{eq:IW_cluster_alloc},
\eqref{eq:BW_cluster_alloc}\text{--}\eqref{eq:BW_non_negativity},
\eqref{eq:total_workload_def}\text{--}\eqref{eq:active_server_limit},
\\[-0.3mm]
&
\eqref{eq:IT_power_cluster}\text{--}\eqref{eq:IT_power_aggregate},
\eqref{eq:grid_power_limit}\text{--}\eqref{eq:battery_discharge_limit},
\eqref{eq:battery_capacity_limit},
\eqref{eq:energy_balance}.
\end{aligned}
\end{gathered}
\end{equation}
}
\subsection{Zero-Sum Perturbation ADMM Framework}

In problem \textbf{P2}, workload migration \eqref{eq:workload_balance} requires coordination among data centers, which exchange their migration volumes $m_{i,t}^I$ to maintain the global workload balance. However, sharing exact migration decisions exposes data centers to operational privacy risks.

To solve \textbf{P2} in a privacy-preserving manner, we develop a zero-sum perturbation ADMM (ZSP-ADMM) scheme that masks the exchanged migration volumes while preserving the global workload balance. Each data center adds a perturbation to its migration volume, and the perturbations are coordinated to sum to zero, leaving the total migration volume unchanged.

Building on this zero-sum property, we reformulate the global
workload balance constraint \eqref{eq:workload_balance} as consensus
constraints on the perturbed migration variables. To determine the
required perturbation magnitude, we first express the coupled migration
decisions as
$\mathbf{m}_t^I=[m_{1,t}^I,\dots,m_{|\mathcal{I}|,t}^I]^T$.
For two adjacent operational states that differ in the workload
information of one data center, the $L_1$-sensitivity of the released
migration decision is defined as

\begingroup
\begin{equation}
\Delta_m
=
\sup_{\mathbf{m}\sim\mathbf{m}'}
\left\lVert \mathbf{m}-\mathbf{m}' \right\rVert_1 .
\end{equation}
\endgroup

The physical migration limits in
\eqref{eq:bandwidth_limit} provide an upper bound on this sensitivity.
Specifically, the maximum admissible change associated with any single
data center is

\begingroup
\[
M_{\max}
=
\max_{i\in\mathcal{I}}
\left(
\Delta m_{i,\mathrm{in}}^{I,\mathrm{limit}}
+
\Delta m_{i,\mathrm{out}}^{I,\mathrm{limit}}
\right).
\]
\endgroup

Therefore, $\Delta_m\leq M_{\max}$, and we conservatively set
$\Delta_m=M_{\max}$ in the perturbation mechanism.

To prevent adversaries from inferring local workload profiles through
shared consensus information, we introduce a structured perturbation
mechanism inspired by \cite{lei2024decentralized}. Let $\epsilon>0$
denote the perturbation-control parameter, which regulates the scale of
the injected noise. At each time slot $t$, prior to the ZSP-ADMM iterations,
each data center $i$ first samples a preliminary perturbation
$\hat{\xi}_{i,t}$ from a zero-mean Laplace distribution with scale
parameter $\Delta_m/\epsilon$, i.e.,
$\hat{\xi}_{i,t}\sim\mathrm{Lap}(\Delta_m/\epsilon)$, whose probability
density function is
$p(\hat{\xi})=\frac{\epsilon}{2\Delta_m}
\exp\left(-\frac{\epsilon|\hat{\xi}|}{\Delta_m}\right)$.

The preliminary perturbations are then coordinated to construct the
zero-sum perturbation
\begin{equation}
\xi_{i,t}
=
\hat{\xi}_{i,t}
-
\frac{1}{|\mathcal{I}|}
\sum_{j\in\mathcal{I}}\hat{\xi}_{j,t},
\end{equation}
which satisfies
$\sum_{i\in\mathcal{I}}\xi_{i,t}=0$.
A fixed positive scaling factor $\zeta$ is used to control the
perturbation amplitude. The obfuscated migration variable is formulated as
\begin{equation}
\tilde{m}_{i,t}^{I}
=
m_{i,t}^{I}
+
\zeta\xi_{i,t},
\label{eq:obfuscated_migration}
\end{equation}
where $\zeta>0$ is a fixed perturbation scaling factor.

Since the coordinated perturbations satisfy
$\sum_{i\in\mathcal I}\xi_{i,t}=0$, the obfuscated migration variables
preserve the original global workload balance, i.e.,
\[
\sum_{i\in\mathcal I}\tilde{m}_{i,t}^{I}
=
\sum_{i\in\mathcal I}m_{i,t}^{I}.
\]
Therefore, the original global migration constraint
\eqref{eq:workload_balance} can be equivalently enforced through the
following consensus constraints on the obfuscated variables:
\begin{equation}\label{eq:node_admm_consensus}
\tilde{m}_{i,t}^I = z_{i,t}^I, \quad \forall i \in \mathcal{I}, \quad \mathbf{z}_t \in \mathcal{Z},
\end{equation}
where $z_{i,t}^I$ is the auxiliary consensus variable for data center $i$ at time slot $t$, and $\mathbf{z}_t = [z_{1,t}^I, \dots, z_{|\mathcal{I}|,t}^I]^T$.
$\mathcal{Z}\triangleq
\{\mathbf{z}_t\in\mathbb{R}^{|\mathcal{I}|}:
\sum_{i\in\mathcal{I}}z_{i,t}^{I}=0\}$
denotes the global workload-balance consensus set.

Replacing the global migration constraint~\eqref{eq:workload_balance} with \eqref{eq:obfuscated_migration} and \eqref{eq:node_admm_consensus} yields the following equivalent formulation of \textbf{P2}, denoted as \textbf{P3}:

{\small
\setlength{\jot}{1pt}
\begin{equation}
\label{eq:p3}
\begin{gathered}
\begin{aligned}
\textbf{P3:}\quad
\min_{\Phi_t,\mathbf z_t}\quad
&
Vf(\Phi_t)
+\sum_{i\in\mathcal I}
\Big[
\delta_i\omega_{i,t}
\big(k_i\tilde E_{i,t}
+\tfrac12\omega_{i,t}
+(1-k_i)\varepsilon_{i,t}\big)
\\[-1.2mm]
&
\qquad
-\alpha_i S^D_{i,t}w^D_{i,t}
+\nu_i\tilde\tau_{i,t}
\big(\kappa_i^B P^{\mathrm{IT}}_{i,t}
-\kappa_i^C P^{\mathrm{cool}}_{i,t}\big)
\Big]
\end{aligned}
\\[0.5mm]
\begin{aligned}
\text{\rm s.t.}\quad &
\eqref{eq:IW_local_balance},
\eqref{eq:bandwidth_limit},
\eqref{eq:IW_cluster_alloc},
\eqref{eq:BW_cluster_alloc}\text{--}\eqref{eq:BW_non_negativity},
\eqref{eq:total_workload_def}\text{--}\eqref{eq:active_server_limit},
\\[-0.3mm]
&
\eqref{eq:IT_power_cluster}\text{--}\eqref{eq:IT_power_aggregate},
\eqref{eq:grid_power_limit}\text{--}\eqref{eq:battery_discharge_limit},
\eqref{eq:battery_capacity_limit},
\eqref{eq:energy_balance},
\eqref{eq:obfuscated_migration},
\eqref{eq:node_admm_consensus}.
\end{aligned}
\end{gathered}
\end{equation}
}
Before deriving the ZSP-ADMM updates, we first establish the convexity of \textbf{P3} and its equivalence to \textbf{P2} under fixed zero-sum perturbations, as stated in the following two lemmas.

\begin{lemma}
For any fixed realization of the zero-sum perturbations, \textbf{P3} is a convex optimization problem.
\end{lemma}

\par\emph{Proof.}
The detailed proof is provided in Appendix~\ref{appendix_B}.

\begin{lemma}
For any fixed zero-sum perturbations, \textbf{P3} and \textbf{P2} have the same feasible physical decisions and optimal objective value.
\end{lemma}

\par\emph{Proof.}
The detailed proof is provided in Appendix~\ref{appendix_C}.

Given the convexity of \textbf{P3}, the detailed update procedure of the proposed ZSP-ADMM is presented below. We begin by constructing the augmented Lagrangian of \textbf{P3}. Let $\mu_{i,t}^I$ denote the dual multiplier associated with the consensus constraint \eqref{eq:node_admm_consensus}, and let $\rho > 0$ denote the penalty parameter. The augmented Lagrangian function is formulated as:
\begingroup
\small
\begin{align}
\label{eq:aug_lagrangian_node}
\mathcal{L}_{\rho}
={}&
\sum_{i\in\mathcal{I}}
\Big[
Vf_i(\Phi_{i,t})
-\alpha_i S^{D}_{i,t}w^{D}_{i,t}
+\nu_i\tilde{\tau}_{i,t}
\big(
\kappa_i^B P^{\mathrm{IT}}_{i,t}
-\kappa_i^C P^{\mathrm{cool}}_{i,t}
\big)
\notag\\[-1.5mm]
&\quad
+\delta_i
\big(
k_i\tilde{E}_{i,t}\omega_{i,t}
+\tfrac{1}{2}\omega_{i,t}^{2}
+(1-k_i)\varepsilon_{i,t}\omega_{i,t}
\big)
\notag\\[-1.5mm]
&\quad
+\mu_{i,t}^{I}
\big(
\tilde{m}_{i,t}^{I}-z_{i,t}^{I}
\big)
+\tfrac{\rho}{2}
\big(
\tilde{m}_{i,t}^{I}-z_{i,t}^{I}
\big)^2
\Big].
\end{align}
\endgroup

In the proposed ZSP-ADMM algorithm, at iteration $n+1$, the decision variables are updated sequentially in three steps.

\textit{1) Parallel Local Primal Update:}
Each data center $i$ independently solves its local subproblem over $\Phi_{i,t}$. The objective minimizes the local augmented Lagrangian decoupled from the consensus constraint:
\begin{align}\label{eq:local_subproblem_p3}
&\Phi^{(n+1)}_{i,t}
= \arg\min_{\Phi_{i,t}} \Big\{ V f_i(\Phi_{i,t}) - \alpha_i S^{D}_{i,t} w^{D}_{i,t} \notag \\
&\qquad\qquad\qquad\quad + \nu_i \tilde{\tau}_{i,t} \big( \kappa_i^B P^{\mathrm{IT}}_{i, t} - \kappa_i^C P^{\mathrm{cool}}_{i, t} \big) + \Omega_{i,t} \Big\} \\
&\text{s.t.} \quad \eqref{eq:IW_local_balance}, \, \eqref{eq:bandwidth_limit}, \, \eqref{eq:IW_cluster_alloc}, \, \eqref{eq:BW_cluster_alloc}\text{--}\eqref{eq:BW_non_negativity}, \, \eqref{eq:total_workload_def}\text{--}\eqref{eq:active_server_limit}, \, \eqref{eq:IT_power_cluster}\text{--}\eqref{eq:IT_power_aggregate}, \notag \\
&\phantom{\text{s.t.} \quad} \eqref{eq:temp_limits}, \, \eqref{eq:grid_power_limit}\text{--}\eqref{eq:battery_discharge_limit}, \, \eqref{eq:battery_capacity_limit}, \, \eqref{eq:energy_balance}, \, \eqref{eq:obfuscated_migration}, \notag
\end{align}
where the decoupling term $\Omega_{i,t}$ collects the local Lyapunov drift and ADMM penalty components:
\begin{align}\label{eq:omega_def}
\Omega_{i,t} &= \delta_i \Big( k_i\tilde{E}_{i,t}\omega_{i,t} + \frac{1}{2}\omega_{i,t}^2 + (1-k_i)\varepsilon_{i,t}\omega_{i,t} \Big) \notag\\
&\quad + \mu_{i,t}^{I,(n)} \left( \tilde{m}_{i,t}^I - z_{i,t}^{I,(n)} \right) + \frac{\rho}{2} \Big( \tilde{m}_{i,t}^I - z_{i,t}^{I,(n)} \Big)^2.
\end{align}
Since the perspective formulation preserves convexity, the local subproblem remains convex and can be solved to global optimality.

\textit{2) Centralized Global Consensus Update:}
The coordinator collects $\tilde{m}_{i,t}^{I,(n+1)}$ from all DCs and updates the global consensus variables $z_{i,t}^I$ by solving the following constrained subproblem:
\begin{align}\label{eq:global_z_update}
\mathbf{z}^{(n+1)}_t &= \arg\min_{\mathbf{z}_t} \sum_{i\in\mathcal{I}} \left[ -\mu_{i,t}^{I,(n)} z_{i,t}^I + \frac{\rho}{2} \Big( \tilde{m}_{i,t}^{I,(n+1)} - z_{i,t}^I \Big)^2 \right] \notag \\
\text{s.t.}\ \ \ & \mathbf{z}_t \in \mathcal{Z}.
\end{align}

\textit{3) Dual Multiplier Update:}
Each data center $i$ updates its dual multiplier locally:
\begin{equation}\label{eq:dual_update_node}
\mu^{I,(n+1)}_{i,t} = \mu^{I,(n)}_{i,t} + \rho \left( \tilde{m}^{I,(n+1)}_{i,t} - z^{I,(n+1)}_{i,t} \right), \quad \forall i \in \mathcal{I}.
\end{equation}

\textit{4) Stopping Criterion:}
The algorithm terminates when the primal residual satisfies
\begin{equation}
\left\|r^{(n)}\right\|_2^2
=
\sum_{i\in\mathcal{I}}
\left(
\tilde{m}^{I,(n)}_{i,t}
-
z^{I,(n)}_{i,t}
\right)^2
\le
\left(\epsilon^{\mathrm{pri}}\right)^2.
\end{equation}

The overall online distributed scheduling procedure based on
ZSP-ADMM is summarized in Algorithm~\ref{alg:distributed_admm}.

\begin{algorithm}[t]
\caption{Online ZSP-ADMM Algorithm}
\label{alg:distributed_admm}

\begingroup
\footnotesize
\linespread{0.92}\selectfont

\begin{algorithmic}[1]
\setlength{\itemsep}{0pt}
\setlength{\parskip}{0pt}
\setlength{\parsep}{0pt}
\setlength{\topsep}{2pt}

\STATE Initialize $S^{D}_{i,0}$, $E_{b,i,0}$, and $\tau_{i,0}$ for all
$i\in\mathcal{I}$.

\FOR{each time slot $t=0,1,\ldots,T-1$}

\STATE Observe the current workload arrivals, renewable generation,
electricity price, carbon intensity, and system states.

\STATE Compute the operational incentive factor $\theta_{i,t}$ and
update the time-varying shift parameter $\varepsilon_{i,t}$ according
to the TVQS mechanism for all $i\in\mathcal{I}$.

\STATE Pre-sample the preliminary perturbation
$\hat{\xi}_{i,t}\sim\mathrm{Lap}(\Delta_m/\epsilon)$
for each DC $i$.

\STATE Construct the coordinated zero-sum perturbation as
$\xi_{i,t}
=
\hat{\xi}_{i,t}
-
\frac{1}{|\mathcal{I}|}
\sum_{j\in\mathcal{I}}\hat{\xi}_{j,t}$,
for all $i\in\mathcal{I}$.

\STATE Form the current per-slot optimization problem \textbf{P3}
using the observed system information, virtual queue states, and
fixed perturbations $\xi_{i,t}$.

\STATE Set ADMM iteration index $n=0$, primal tolerance
$\epsilon^{\mathrm{pri}}>0$, and penalty parameter $\rho>0$.
Initialize $\mu_{i,t}^{I,(0)}=0$ and $z_{i,t}^{I,(0)}=0$
for all $i\in\mathcal{I}$.

\REPEAT

    \FOR{each DC $i\in\mathcal{I}$}

        \STATE Solve the local subproblem
        \eqref{eq:local_subproblem_p3} to obtain
        $\Phi_{i,t}^{(n+1)}$.

        \STATE Compute and transmit the obfuscated migration variable
        $\tilde{m}_{i,t}^{I,(n+1)}
        =
        m_{i,t}^{I,(n+1)}
        +
        \zeta\xi_{i,t}$.

    \ENDFOR

    \STATE The coordinator updates $\mathbf{z}_{t}^{(n+1)}$
    according to \eqref{eq:global_z_update}.

    \STATE Each DC $i\in\mathcal{I}$ updates
    $\mu_{i,t}^{I,(n+1)}$ according to
    \eqref{eq:dual_update_node}.

    \STATE $n\leftarrow n+1$.

\UNTIL{the primal residual satisfies the prescribed convergence tolerance}

\STATE Implement the converged control decisions $\Phi_t$.

\STATE Update $S^{D}_{i,t+1}$, $E_{b,i,t+1}$, and $\tau_{i,t+1}$
according to \eqref{eq:batch_queue}, \eqref{eq:battery_soc},
and \eqref{eq:thermal_dynamic}, respectively.

\ENDFOR

\end{algorithmic}

\endgroup
\vspace{-0.2em}
\end{algorithm}

\section{Numerical Results}

In this section, we evaluate our proposed online distributed algorithm based on real-world traces. We first detail the simulation setup and parameter configurations. Subsequently, we validate the effectiveness of the proposed model and algorithm.

\par\medskip
\noindent\textit{A. Parameter Setup}\quad

\begin{figure}[htbp]
    \centering
    \includegraphics[width=1\linewidth]{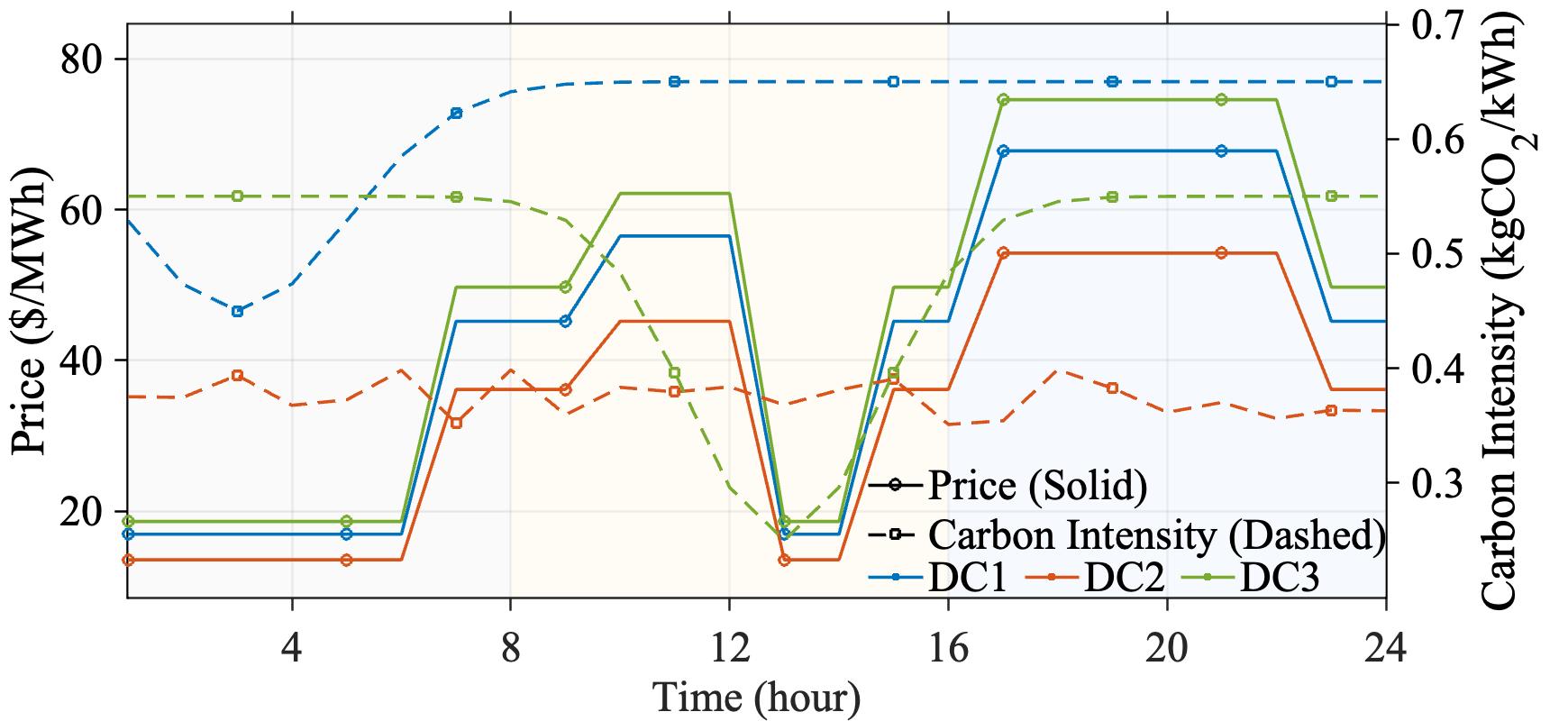}
    \caption{Electricity prices and carbon intensities.}
    \label{fig_param}
\end{figure}
We consider three interconnected DCs located in regions with heterogeneous energy resources. DC1 is wind-dominated with lower carbon intensity during nighttime, DC2 is hydro-based with consistently low carbon emissions and electricity prices, and DC3 features high PV penetration, offering a decarbonization window around midday. The hourly profiles of electricity prices $\lambda_{i,t}$ and carbon intensities $\rho_{i,t}^{oc}$ are shown in Fig.~\ref{fig_param}.

\begin{table}[h]
\centering
\caption{Parameters of Heterogeneous Server Clusters}
\label{table_cluster}
\renewcommand{\arraystretch}{1.2}
\begin{tabular}{l c c c c}
\hline
\textbf{Cluster Index} 
& \textbf{Cluster~1} 
& \textbf{Cluster~2} 
& \textbf{Cluster~3} 
& \textbf{Cluster~4} \\
\hline
Computing Capacity & Low & High & Low & High \\
Service Age & New & New & Old & Old \\
\hline
$S^{\mathrm{rate}}_{i,k}$ (req/s) & 20 & 30 & 20 & 30 \\
$P_{\mathrm{idle}}$ (kW) & 0.100 & 0.125 & 0.125 & 0.156 \\
$P_{\mathrm{peak}}$ (kW) & 0.200 & 0.250 & 0.250 & 0.312 \\
$\gamma_{i,k}$ & 1.00 & 1.10 & 2.00 & 2.20 \\
\hline
\end{tabular}
\end{table}
Each DC comprises four server clusters, with 1000 servers in each cluster, distinguished by computing capacity and service age, as summarized in Table~\ref{table_cluster}. The global simulation parameters are listed in Table~\ref{table_global}.
\begin{table}[t]
\centering
\caption{Global Simulation Parameters}
\label{table_global}
\renewcommand{\arraystretch}{1.15}
\begin{tabular}{ll|ll}
\hline
\textbf{Symbol} & \textbf{Value} & \textbf{Symbol} & \textbf{Value} \\
\hline
$\kappa_i^B$ & 0.002
& $\kappa_i^C$ & 0.008 \\
$\Delta m^{I,\mathrm{limit}}_{i,\mathrm{in}},
\Delta m^{I,\mathrm{limit}}_{i,\mathrm{out}}$
& $1.5 \times 10^4$ req
& $c_i^{\mathrm{trans}}$
& $10^{-4}$ \$/req \\
$E_b$ & 4,000 kWh
& $E_{b,i}^{\min}$ & $0.1E_b$ \\
$P_{b,i,t}^{c,\max}$ & 400 kW
& $E_{b,i}^{\max}$ & $0.9E_b$ \\
$P_{b,i,t}^{d,\max}$ & 400 kW
& $\Delta E_{b,i}$ & 400 kWh \\
$\eta_c$ & 0.95
& $k_i$ & 0.98 \\
$\eta_d$ & 0.95
& $C^M_{i,k}$ & 1200 kg \\
$\psi^{\mathrm{deg}}$ & 0.01 \$/kWh
& $\alpha_C$ & 0.1 \$/kgCO$_2$ \\
$\bar{t}^{\mathrm{delay}}_{i,k}$ & 0.167 s
& $L^{\mathrm{life}}_{i,k}$ & 35,040 h \\
\hline
\end{tabular}
\end{table}

\par\medskip
\noindent\textit{B. Performance Comparison}\quad

To evaluate the impact of workload migration and hardware aging awareness, four scenarios are compared:

\begin{itemize}
    \item \textbf{Benchmark}: No workload migration; no server aging model.
    \item \textbf{M2}: Workload migration enabled; no server aging model.
    \item \textbf{M3}: No workload migration; server aging model incorporated.
    \item \textbf{Proposed}: Both workload migration and server aging model enabled.
\end{itemize}

\begin{table}[htbp]
\centering
\caption{Comparison of Economic Cost and Total Carbon Emissions}
\label{tab:comparison}
\begin{tabular}{ccccc}
\toprule
 & Benchmark & M2 & M3 & Proposed \\
\midrule
Cost ($\times 10^3$ \$) & 26.84 & 24.97 & 24.93 & 23.47 \\
Cost Reduction & - & 7.0\% & 7.1\% & 12.6\% \\
Carbon Emission (t CO$_2$) & 170.07 & 156.54 & 157.00 & 148.04 \\
Carbon Reduction & - & 8.0\% & 7.7\% & 13.0\% \\
\bottomrule
\end{tabular}
\end{table}
Table~\ref{tab:comparison} compares the performance of different methods. Compared with the benchmark, M2 reduces total cost by 7.0\% and carbon emissions by 8.0\%. M2 transfers workloads to data centers with lower electricity prices and grid carbon intensities, thereby reducing electricity costs and operational carbon emissions. M3 reduces total cost by 7.1\% and carbon emissions by 7.7\%. Its aging model gives priority to server clusters with lower degradation rates, reducing hardware aging and replacement-induced embodied carbon emissions.
The proposed method jointly optimizes inter-data-center workload transfer and aging-aware server allocation. It transfers workloads to data centers with lower electricity prices and grid carbon intensities while prioritizing server clusters with lower degradation rates within each data center. Consequently, it achieves the lowest total cost (23.47 $\times 10^3$ \$) and carbon emissions (148.04 t CO$_2$), representing reductions of 12.6\% and 13.0\%, respectively.

\begin{figure}[h]
    \centering
    \includegraphics[width=1\linewidth,
    trim={50bp 5bp 40bp 50bp},
        clip]{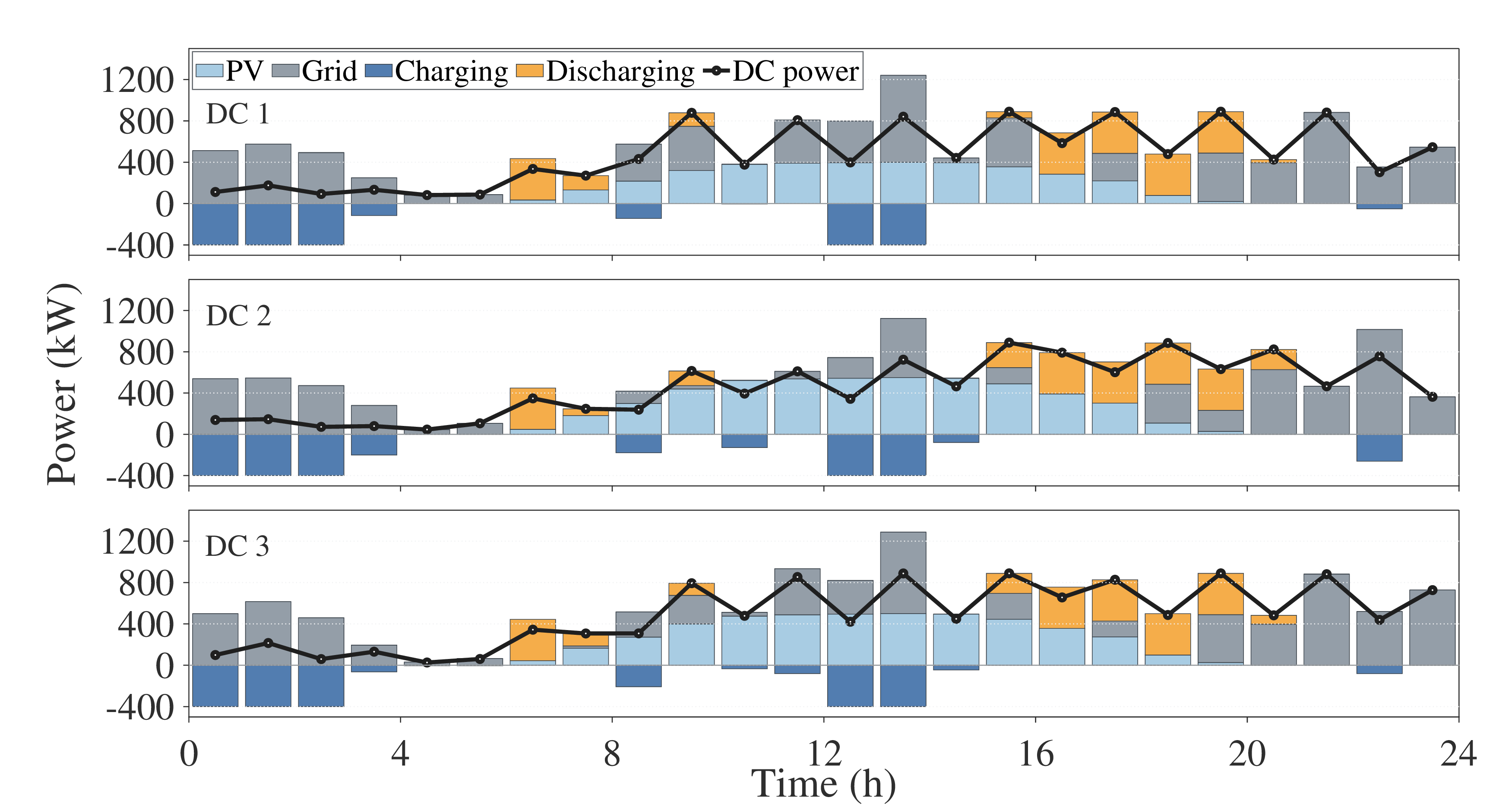}
  \caption{Solar power generation, grid power supply, battery
operation, and DC power consumption across the three data centers.}
    \label{fig_param1}
\end{figure}

\par\medskip
\noindent\textit{C. Operational Scheduling Analysis}\quad

 The power composition of each data center within this typical day is shown in Fig.~\ref{fig_param1}. During morning and midday periods, relatively low electricity prices and abundant PV generation encourage battery charging. Conversely, during evening peak hours, limited renewable supply, high electricity prices, and increased workload demand drive battery discharging.

Fig.~\ref{fig_number_comparison} shows the number of servers utilized
in each cluster across the three data centers over 24 hours. During the late-night and early-morning hours, workload demand is low, and the scheduler mainly uses Cluster~2, with limited use of Cluster~1. Both clusters age slowly, but Cluster~2 offers a higher processing rate and is therefore prioritized. During the daytime and evening peak periods, additional capacity is needed to process accumulated batch workloads and increased interactive demand. Cluster~3 and Cluster~4 are therefore used more extensively despite their higher aging rates. Their utilization decreases again when demand falls, limiting the use of servers that degrade faster.

Fig.~\ref{time} shows that Cluster~1 and Cluster~2 accumulate longer equivalent operating times than Cluster~3 and Cluster~4 across all three data centers. This is consistent with the utilization patterns in Fig.~\ref{fig_number_comparison} and confirms the scheduler's overall preference for clusters with lower aging rates.

\begin{figure}[h]
    \centering
    \includegraphics[
        width=\linewidth,
        trim={185bp 5bp 225bp 5bp},
        clip
    ]{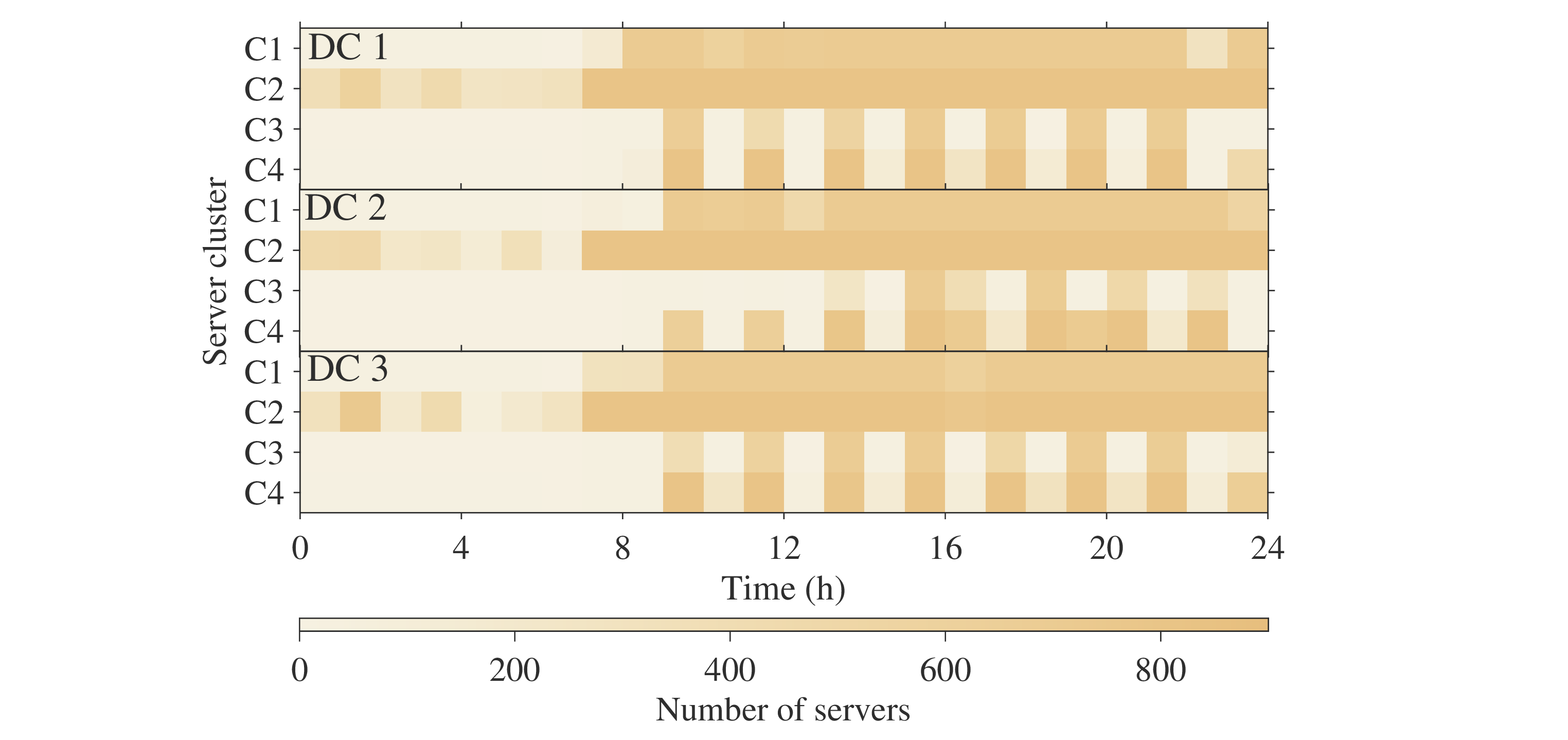}
    \caption{Number of servers utilized} in each cluster across the
    three data centers over 24 hours.
    \label{fig_number_comparison}
\end{figure}
\begin{figure}[h]
    \centering
    \includegraphics[
        width=\linewidth,
        trim={25bp 5bp 45bp 65bp},
        clip
    ]{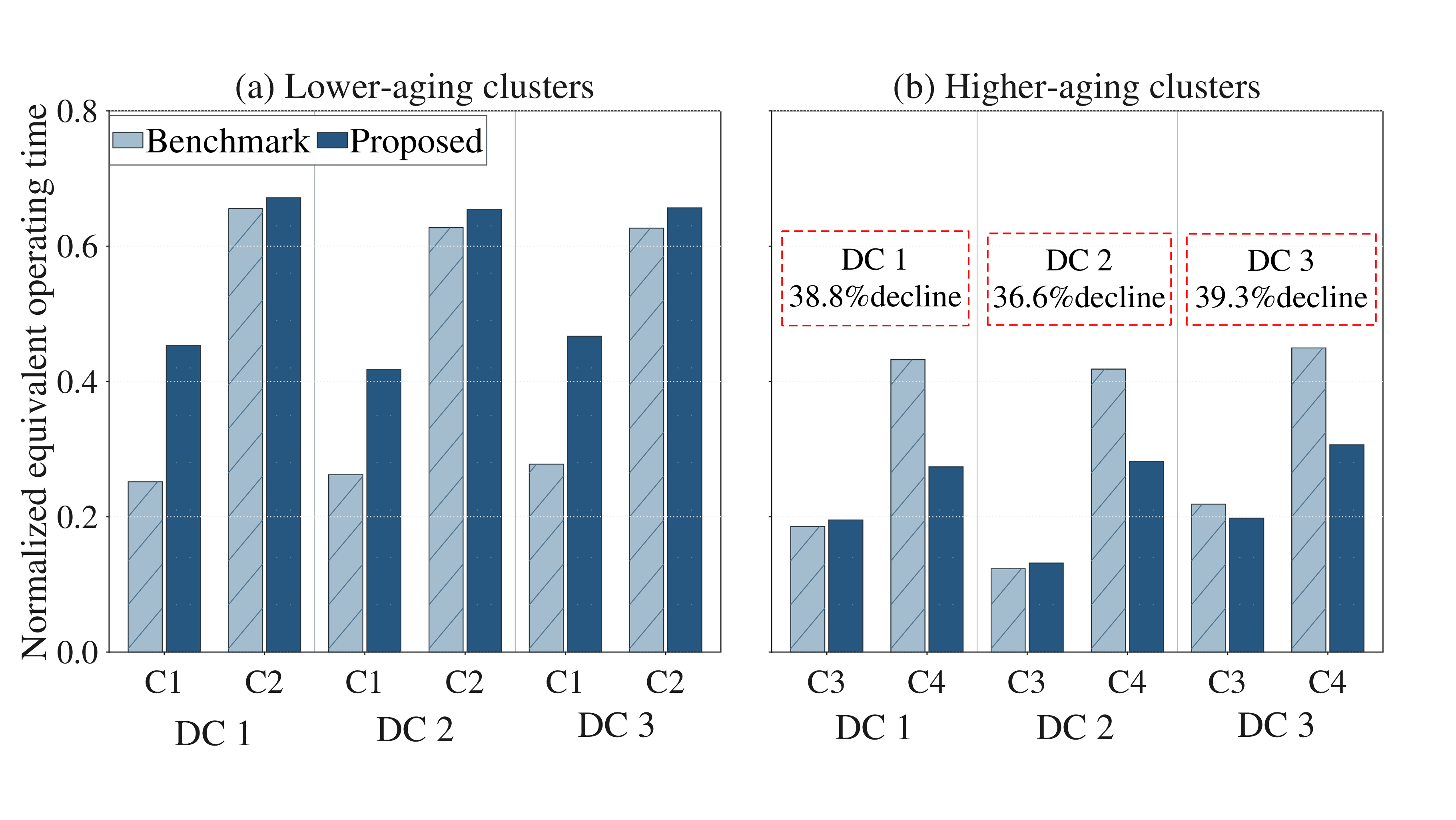}
  \caption{Normalized equivalent operating times (EOTs) of different server clusters under the benchmark and proposed methods.}
    \label{time}
\end{figure}

\par\medskip
\noindent\textit{D. Impact of TVQS on Battery Scheduling}\quad

Fig.~\ref{battery} compares the battery charging and discharging power
under fixed queue shifting and the proposed TVQS mechanism. Fixed queue
shifting results in relatively weak and scattered battery operation,
indicating conservative battery utilization. In contrast, TVQS promotes
charging during low-price or renewable-abundant periods, particularly
during the early hours and around hours 13--14, and releases the stored
energy during the evening high-price period. This pattern is driven by
the time-varying queue shift. An increase in renewable generation or a
decrease in the combined electricity and carbon price makes the shift
more negative, thereby favoring charging. Conversely, the opposite
conditions move the shift toward zero and favor discharging.
\begin{figure}[h]
    \centering
    \includegraphics[width=\linewidth,
    trim={5bp 5bp 5bp 5bp},
        clip]{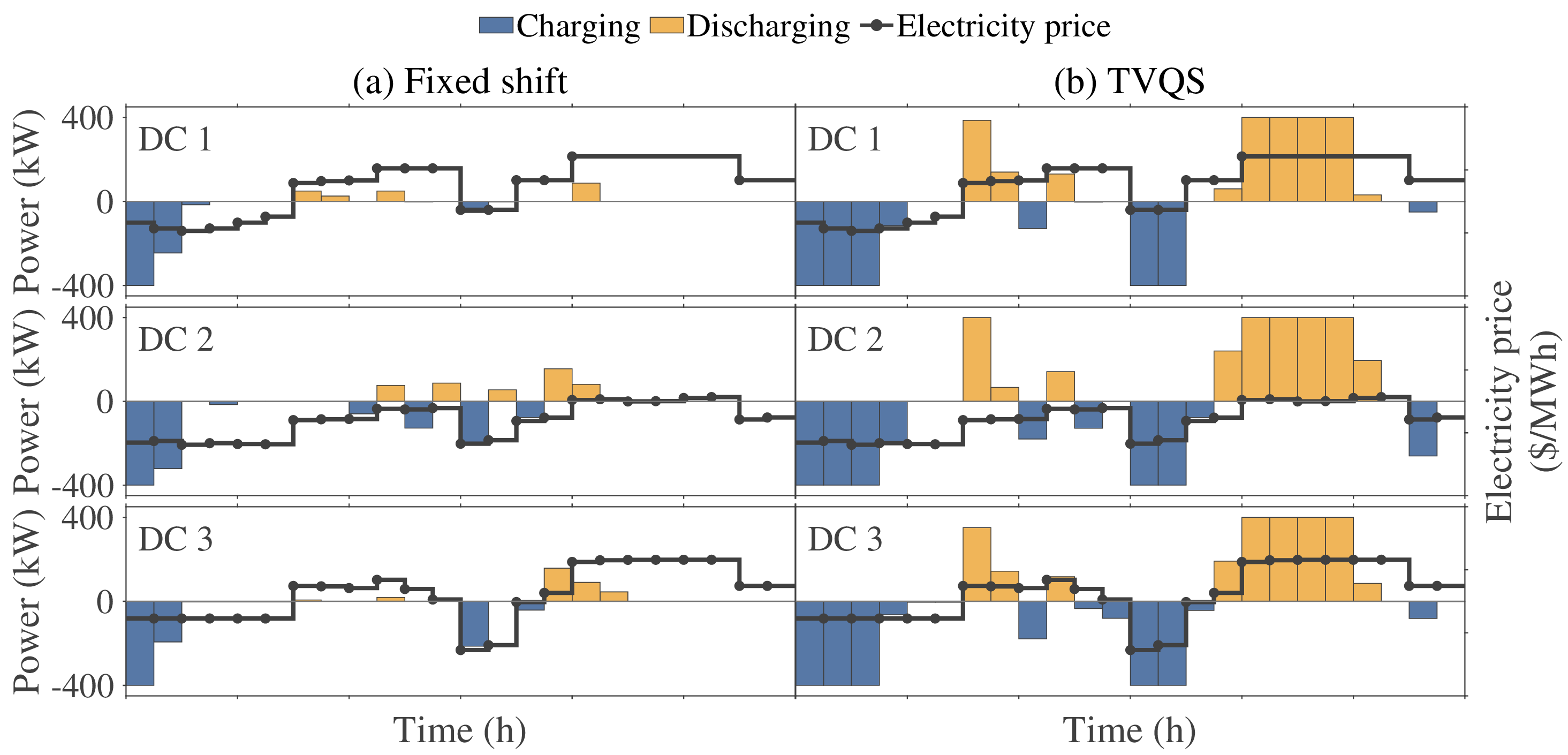}
    \caption{Battery charging and discharging power under fixed queue
    shifting and the proposed TVQS mechanism.}
    \label{battery}
\end{figure}

\par\vspace{1mm}
\noindent\textit{E. Convergence and Privacy-Preservation Performance of
ZSP-ADMM}\quad

Fig.~\ref{fig:admm} compares ZSP-ADMM under different perturbation-control parameters
($\epsilon\in\{0.2,0.3,0.4,0.5\}$) with standard ADMM.
All curves closely track the baseline and converge within 70 iterations,
indicating that the structured perturbation has little impact on
convergence.

Fig.~\ref{fig:transfer} compares the actual and obfuscated inter-data-center
migration signals over 24 hours. The insets report the mean absolute
error (MAE), root mean square error (RMSE), and maximum absolute
deviation. Across the three data centers, the MAE and RMSE remain within
$0.17$--$0.21$ and $0.22$--$0.26$, respectively, while the maximum
deviation is $0.43$ (all in units of $10^4$ requests). These bounded
deviations obscure the exact migration values while preserving their
overall temporal patterns for distributed coordination.

\begin{figure}[h]
    \centering
    \includegraphics[
        width=\linewidth,
        trim={32pt 5pt 40pt 18pt},
        clip
    ]{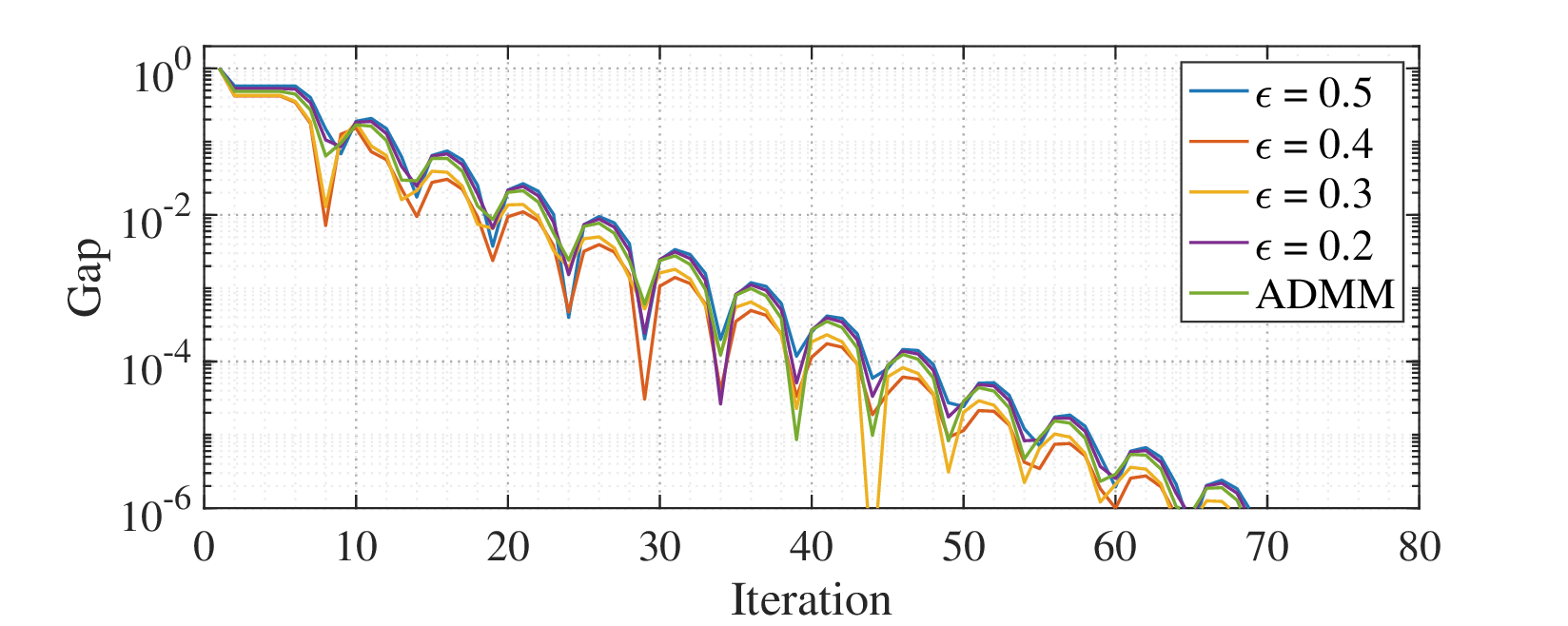}
    \caption{Convergence trajectories of the proposed ZSP-ADMM under
    different perturbation-control parameters.}
    \label{fig:admm}
\end{figure}

\begin{figure}[h]
    \centering
    \includegraphics[
        width=\linewidth,
        trim={25pt 8pt 30pt 55pt},
        clip
    ]{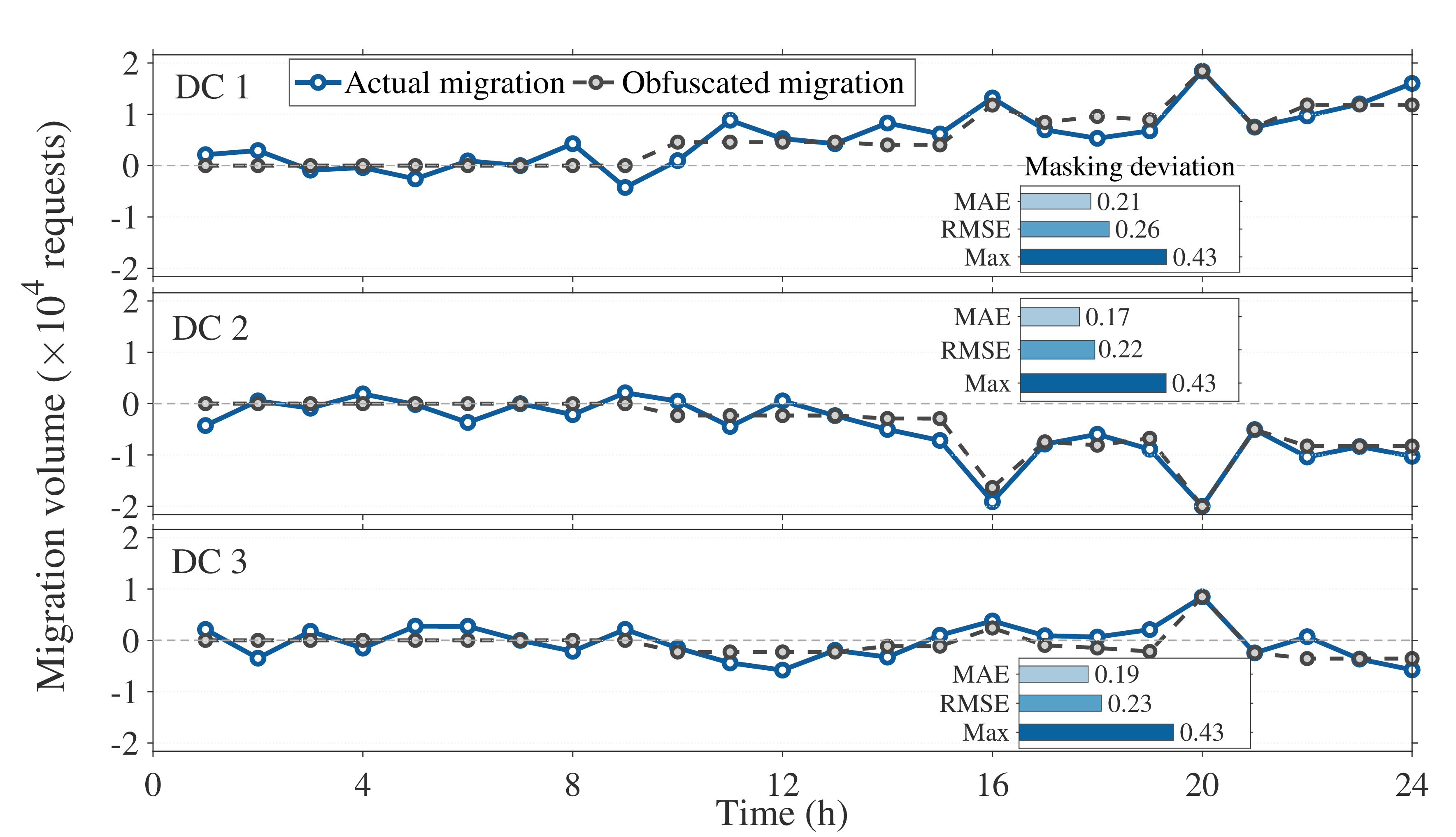}
    \caption{Privacy-preserving workload masking performance on
    transmitted migration signals.}
    \label{fig:transfer}
\end{figure}
\section{Conclusion}
This paper proposes an aging-aware online distributed scheduling framework for geo-distributed data centers. The model incorporates utilization-dependent hardware aging into lifecycle carbon accounting. A Lyapunov-based method with TVQS enables online scheduling and adaptive battery management, while ZSP-ADMM masks migration volumes during distributed coordination. Simulation results yield the following findings: 1) Combining workload migration with aging-aware scheduling reduces economic costs and lifecycle carbon emissions while preserving hardware health by prioritizing clusters with lower aging rates. 2) The proposed Lyapunov-based strategy with TVQS reduces conservative battery utilization by adapting charging and discharging decisions to real-time renewable generation and electricity prices. 3) Zero-sum perturbations mask local migration decisions while preserving global workload balance, with little change in the ZSP-ADMM convergence behavior.

\bibliographystyle{ieeetr} 
\bibliography{refs}       
\section*{Appendices}
\appendices

\section{Proof of the Lyapunov Drift Upper Bound} \label{appendix_A}

The conditional Lyapunov drift is defined as
$\Delta(\boldsymbol{\Theta}_t)
=
\mathbb{E}[L(\boldsymbol{\Theta}_{t+1}) - L(\boldsymbol{\Theta}_t)
\mid \boldsymbol{\Theta}_t]$.
We bound $\Delta(\boldsymbol{\Theta}_t)$ by squaring the dynamic transition equations of each queue independently.

\textbf{1) Batch Workload Queue:}
Squaring the queue dynamics
$S^D_{i,t+1}
=
\max\{S^D_{i,t}-w^D_{i,t},0\}+A^D_{i,t}$
yields:
\begin{equation}
    \frac{\alpha_i}{2}
    \left[
    (S^{D}_{i,t+1})^2 - (S^{D}_{i,t})^2
    \right]
    \le
    \alpha_i S^{D}_{i,t}
    (A^{D}_{i,t} - w^{D}_{i,t})
    + C_{S,i},
\end{equation}
where
$C_{S,i}
=
\frac{\alpha_i}{2}
\max\{(A_i^{D,\max})^2, (w_i^{D,\max})^2\}$
is a finite constant bounding the higher-order terms.

\textbf{2) Thermal Virtual Queue:}
Squaring the thermal transition
$\tilde{\tau}_{i,t+1}
=
\tilde{\tau}_{i,t}
+
\kappa_i^B P^{\mathrm{IT}}_{i,t}
-
\kappa_i^C P^{\mathrm{cool}}_{i,t}
-
\beta_{i,t}$
yields:
\begin{equation}
\begin{aligned}
    \frac{\nu_i}{2}
    \left[
    \tilde{\tau}_{i,t+1}^2
    -
    \tilde{\tau}_{i,t}^2
    \right]
    \le\;
    &\nu_i \tilde{\tau}_{i,t}
    \big(
    \kappa_i^B P^{\mathrm{IT}}_{i,t}
    -
    \kappa_i^C P^{\mathrm{cool}}_{i,t}
    \big) \\
    &-
    \nu_i \tilde{\tau}_{i,t}\beta_{i,t}
    +
    C_{\tau,i},
\end{aligned}
\end{equation}
where
$C_{\tau,i}
=
\frac{\nu_i}{2}
\max\{
(\kappa_i^B P_i^{\mathrm{IT,\max}})^2,
(\kappa_i^C P_i^{\mathrm{cool,\max}} + \beta_i^{\max})^2
\}$.
Based on \eqref{eq:temp_limits}, the absolute temperature deviation is strictly bounded by
\[
\tilde{\tau}_i^{\max}
=
\max\{
|\overline{T}_i - T_i^{\mathrm{ref}}|,
|\underline{T}_i - T_i^{\mathrm{ref}}|
\}.
\]
Thus, the control-independent cross-term
$-\nu_i \tilde{\tau}_{i,t}\beta_{i,t}$
has a finite upper bound
$C_{\beta,i}
=
\nu_i \tilde{\tau}_i^{\max}\beta_i^{\max}$.
Let
$C_{\tau,i}' = C_{\tau,i} + C_{\beta,i}$.
The thermal drift bound simplifies to:
\begin{equation}
    \frac{\nu_i}{2}
    \left[
    \tilde{\tau}_{i,t+1}^2
    -
    \tilde{\tau}_{i,t}^2
    \right]
    \le
    \nu_i \tilde{\tau}_{i,t}
    \big(
    \kappa_i^B P^{\mathrm{IT}}_{i,t}
    -
    \kappa_i^C P^{\mathrm{cool}}_{i,t}
    \big)
    +
    C_{\tau,i}'.
\end{equation}

\textbf{3) Battery Virtual Queue:}
Under the TVQS strategy, the shift parameter is time-varying. Let
$\Delta\varepsilon_{i,t}
=
\varepsilon_{i,t+1}
-
\varepsilon_{i,t}$.
The exact virtual queue transition is
\[
\tilde{E}_{i,t+1}
=
k_i \tilde{E}_{i,t}
+
\omega_{i,t}
+
(1-k_i)\varepsilon_{i,t}
+
\Delta\varepsilon_{i,t}.
\]
Squaring this transition yields:
\begin{equation}\label{eq:drift_expansion_energy}
\begin{aligned}
    \frac{\delta_i}{2}
    \left[
    \tilde{E}_{i,t+1}^2
    -
    \tilde{E}_{i,t}^2
    \right]
    &=
    \delta_i
    \Big[
    k_i \tilde{E}_{i,t}\omega_{i,t}
    +
    \frac{1}{2}\omega_{i,t}^2 \\
    &\qquad
    +
    (1-k_i)\varepsilon_{i,t}\omega_{i,t}
    \Big]
    +
    M_{i,t}
    +
    \Theta_{\varepsilon,i,t},
\end{aligned}
\end{equation}
where $M_{i,t}$ represents the bounding terms inherited from the physical dynamics:
\begin{equation}
    M_{i,t}
    =
    \frac{\delta_i}{2}
    \Big[
    (k_i^2 - 1)\tilde{E}_{i,t}^2
    +
    2k_i(1-k_i)\tilde{E}_{i,t}\varepsilon_{i,t}
    +
    (1-k_i)^2\varepsilon_{i,t}^2
    \Big].
\end{equation}

The term $\Theta_{\varepsilon,i,t}$ captures the additional terms introduced by the time variation of the shift parameter:
\begin{equation}
    \Theta_{\varepsilon,i,t}
    =
    \frac{\delta_i}{2}
    \Big[
    2
    \big(
    k_i \tilde{E}_{i,t}
    +
    \omega_{i,t}
    +
    (1-k_i)\varepsilon_{i,t}
    \big)
    \Delta\varepsilon_{i,t}
    +
    (\Delta\varepsilon_{i,t})^2
    \Big].
\end{equation}

The physical states and parameters strictly satisfy
$|\tilde{E}_{i,t}| \le \tilde{E}_i^{\max}$,
$|\omega_{i,t}| \le \omega_i^{\max}$, and
$|\Delta\varepsilon_{i,t}|
\le
\varepsilon_i^{\max}
-
\varepsilon_i^{\min}$,
where
\[
\tilde{E}_i^{\max}
\triangleq
\max
\left\{
\left|E_{b,i}^{\min}+\varepsilon_i^{\min}\right|,
\left|E_{b,i}^{\max}+\varepsilon_i^{\max}\right|
\right\}.
\]
Therefore, $M_{i,t}$ and $\Theta_{\varepsilon,i,t}$ are uniformly bounded over all feasible control actions. Hence, there exists a finite constant $C_{E,i}$ such that
\begin{equation}
    M_{i,t} + \Theta_{\varepsilon,i,t}
    \le
    C_{E,i}.
\end{equation}

\textbf{4) Total Drift Upper Bound:}
Taking the conditional expectation of the summed individual queue bounds yields the total drift. Let
\[
C_0
=
\sum_{i \in \mathcal{I}}
\left(
C_{S,i}
+
C_{\tau,i}'
+
C_{E,i}
\right)
\]
be the aggregated finite constant.
The conditional one-slot Lyapunov drift is bounded by:
\begin{equation}
\begin{aligned}
\Delta(\boldsymbol{\Theta}_t)
\le
C_0
+
\sum_{i \in \mathcal{I}}
\Big\{
&\alpha_i S^{D}_{i,t}
\big(
A^{D}_{i,t}
-
w^{D}_{i,t}
\big) \\
&\hspace{-2.5em}
+\nu_i \tilde{\tau}_{i,t}
\big(
\kappa_i^B P^{\mathrm{IT}}_{i,t}
-
\kappa_i^C P^{\mathrm{cool}}_{i,t}
\big) \\
&\hspace{-2.5em}
+\delta_i
\big[
k_i\tilde{E}_{i,t}\omega_{i,t}
+ \frac{1}{2}\omega_{i,t}^2
+ (1-k_i)\varepsilon_{i,t}\omega_{i,t}
\big]
\Big\}.
\end{aligned}
\end{equation}
This completes the proof of the stated upper bound on the conditional one-slot Lyapunov drift.

\section{Proof of the convexity of the problem}
\label{appendix_B}

This section establishes the convexity of \textbf{P3} for any fixed
realization of the zero-sum perturbations.
The active-server variables are treated as continuous.
At each time slot, the current system states, exogenous inputs,
queue shifts, and perturbation scaling factor are fixed.
We establish the convexity of both the objective and the feasible region.

\par\noindent
\textit{1) Convexity of the objective:}
Considering the contribution of cluster $(i,k)$ to the aging-amortized
embodied carbon cost $f^{c,m}_{i,t}$, and substituting the utilization
definition
$u_{i,k,t}=W_{i,k,t}/(N^{A}_{i,k,t}S^{\mathrm{rate}}_{i,k})$,
the corresponding per-cluster cost can be written as
\begin{equation}
f^{c,m}_{i,k,t}(W_{i,k,t},N^{A}_{i,k,t})
=
A_{i,k}N^{A}_{i,k,t}
\exp\left(
B_{i,k}\frac{W_{i,k,t}}{N^{A}_{i,k,t}}
\right),
\end{equation}
where
$A_{i,k}
=
\alpha_C C^{M}_{i,k}\eta_{i,k}\Delta t/L^{\mathrm{life}}_{i,k}>0$
and
$B_{i,k}
=
\gamma_{i,k}/S^{\mathrm{rate}}_{i,k}>0$
are positive constants.

We define the base function $h(x)=\exp(B_{i,k}x)$, which is convex
for $x\geq0$. According to convex optimization theory
\cite{vandenberghe2004convex}, the perspective function
$g(x,y)=yh(x/y)$ preserves joint convexity with respect to
$(x,y)$ on the domain $y>0$.

By letting $x=W_{i,k,t}$ and $y=N^{A}_{i,k,t}$,
$f^{c,m}_{i,k,t}$ is the perspective function of $A_{i,k}h(x)$:
\begin{equation}
f^{c,m}_{i,k,t}(x,y)
=
A_{i,k}yh\left(\frac{x}{y}\right).
\end{equation}
The lower bound $N^{A}_{i,k,t}\geq N^{\min}_{i,k}>0$
ensures that this perspective is well defined.
Thus, the embodied carbon cost is jointly convex in
$(W_{i,k,t},N^{A}_{i,k,t})$.
Since $W_{i,k,t}=w^I_{i,k,t}+w^D_{i,k,t}$ is affine,
convexity is preserved in the original allocation variables.

The electricity purchase, battery cycling, and operational
carbon costs are affine, while the transfer cost
$c_i^{\mathrm{trans}}\max(m^I_{i,t},0)$ is convex for
$c_i^{\mathrm{trans}}\geq0$.
The quadratic drift term $\delta_i\omega_{i,t}^2/2$ is convex
for $\delta_i\geq0$, since
$\omega_{i,t}=P^c_{b,i,t}\Delta t\eta_c
-P^d_{b,i,t}\Delta t/\eta_d$ is affine.
All remaining drift terms are affine for fixed queue states
and shifts. Therefore, the complete objective is convex for $V>0$.

\par\noindent
\textit{2) Convexity of the feasible region:}
\begingroup
Using the utilization definition, the delay constraint
\eqref{eq:SLA_delay_limit}, on its positive-denominator domain,
is equivalently expressed as
\begin{equation}
W_{i,k,t}
\leq
\left(
S^{\mathrm{rate}}_{i,k}
-\frac{1}{\bar{t}^{\mathrm{delay}}_{i,k}}
\right)N^{A}_{i,k,t}.
\end{equation}
This is an affine inequality because the processing rate and
delay limit are fixed positive parameters.
Similarly, the utilization bounds become
$0\leq W_{i,k,t}\leq
N^{A}_{i,k,t}S^{\mathrm{rate}}_{i,k}$.

The computing power relation \eqref{eq:IT_power_cluster}
can be rewritten as
\begin{equation}
\begin{aligned}
P^{\mathrm{IT}}_{i,k,t}
={}&P^{\mathrm{idle}}_{i,k}N^{A}_{i,k,t}
+\frac{P^{\mathrm{peak}}_{i,k}-P^{\mathrm{idle}}_{i,k}}
{S^{\mathrm{rate}}_{i,k}}W_{i,k,t},
\end{aligned}
\end{equation}
which is affine. The workload allocation, aggregate computing
power, and DC power-balance relations are also affine.
The remaining migration, nonnegativity, server-capacity,
power, and state bounds define affine inequalities for the
fixed per-slot parameters.

For fixed $\xi_{i,t}$ and $\zeta$, the obfuscated migration
relation \eqref{eq:obfuscated_migration} and consensus constraints
\eqref{eq:node_admm_consensus} take the affine form
\[
\tilde{m}^{I}_{i,t}=m^{I}_{i,t}+\zeta\xi_{i,t},
\qquad
\tilde{m}^{I}_{i,t}=z^{I}_{i,t},
\qquad
\sum_{i\in\mathcal{I}}z^{I}_{i,t}=0.
\]
Consequently, the feasible region is an intersection of affine
equality and inequality sets and is convex.
Together with the convexity of the objective, this establishes
that \textbf{P3} is a convex optimization problem.
\par
\endgroup

\section{Proof of the Equivalence Between P2 and P3}
\label{appendix_C}

\begingroup

Fix any zero-sum perturbation satisfying
$\sum_{i\in\mathcal I}\xi_{i,t}=0$.
The two problems impose the same local constraints on the physical
decisions $\Phi_t$. We show that their remaining constraints are
equivalent.

For any feasible solution $\Phi_t$ of \textbf{P2}, its migration
decisions satisfy $\sum_{i\in\mathcal I}m_{i,t}^{I}=0$.
Set
\[
\tilde m_{i,t}^{I}
=z_{i,t}^{I}
=m_{i,t}^{I}+\zeta\xi_{i,t}.
\]
Then
\[
\sum_{i\in\mathcal I}z_{i,t}^{I}
=
\sum_{i\in\mathcal I}m_{i,t}^{I}
+\zeta\sum_{i\in\mathcal I}\xi_{i,t}
=0.
\]
Thus, $\Phi_t$ can be extended to a feasible solution of
\textbf{P3}.

Conversely, for any feasible solution of \textbf{P3},
$\tilde m_{i,t}^{I}=z_{i,t}^{I}
=m_{i,t}^{I}+\zeta\xi_{i,t}$ and
$\sum_{i\in\mathcal I}z_{i,t}^{I}=0$ imply
\[
\sum_{i\in\mathcal I}m_{i,t}^{I}
=
\sum_{i\in\mathcal I}z_{i,t}^{I}
-\zeta\sum_{i\in\mathcal I}\xi_{i,t}
=0.
\]
Hence, its physical decisions $\Phi_t$ are feasible for
\textbf{P2}. The two problems therefore have the same feasible
physical decisions.

Finally, the objective functions of \textbf{P2} and \textbf{P3}
are identical for the same $\Phi_t$ and do not depend on
$\mathbf z_t$. Every corresponding pair of feasible solutions
thus has the same objective value. Consequently, \textbf{P2}
and \textbf{P3} have the same optimal objective value and
corresponding optimal physical decisions.

\par
\endgroup
\end{document}